\documentclass[12pt]{amsart}
 
\usepackage[margin = 1in]{geometry}
\usepackage{amsmath,amsfonts,amssymb,amsthm}
\usepackage{color}
\usepackage{hyperref}
\usepackage{graphicx}
\usepackage{tikz}
\numberwithin{equation}{section}
 
\theoremstyle{theorem}
   \newtheorem{theorem}[equation]{Theorem}
   
   \newtheorem{lemma}[equation]{Lemma}
   
   \newtheorem{proposition}[equation]{Proposition}
   \newtheorem{corollary}[equation]{Corollary}

\theoremstyle{definition}

\theoremstyle{remark}
   \newtheorem{remark}[equation]{Remark}

\newcommand{\R}{\mathbb{R}}
\newcommand{\N}{\mathbb{N}}
\newcommand{\E}{\mathbb{E}}
\newcommand{\X}{\mathbb{X}}
\newcommand{\tensor}{\otimes}

\newcommand\new[1]{#1}

\DeclareMathOperator{\atan}{atan}
\DeclareMathOperator{\ins}{ins}
\DeclareMathOperator{\pin}{pin}

\title{Persistent homology detects curvature}

\author{Peter Bubenik}
\address{Department of Mathematics, University of Florida}
\email{peter.bubenik@ufl.edu}
\urladdr{https://people.clas.ufl.edu/peterbubenik/}

\author{Michael Hull}
\address{Department of Mathematics \& Statistics, University of North Carolina at Greensboro}
\email{mbhull@uncg.edu}
\urladdr{https://mathstats.uncg.edu/people/directory/michael-hull/}

\author{Dhruv Patel}
\address{Department of Statistics, University of North Carolina -- Chapel Hill}
\email{dhruvpat@live.unc.edu}

\author{Benjamin Whittle}


\begin{document}

\maketitle

\begin{abstract}
  In topological data analysis, persistent homology is used to study the ``shape of data''.
  Persistent homology computations are completely characterized by a set of intervals called a bar code. It is often said that the long intervals represent the ``topological signal'' and the short intervals represent ``noise''. We give evidence to dispute this thesis, showing that the short intervals encode  geometric information. Specifically, we prove that persistent homology detects the curvature of disks from which points have been sampled. We describe a general computational framework for solving inverse problems using the average persistence landscape, a continuous mapping from metric spaces with a probability measure to a Hilbert space. In the present application, the average persistence landscapes of points sampled from disks of constant curvature results in a path in this Hilbert space which may be learned using standard tools from statistical and machine learning.
\end{abstract}

\section{Introduction} \label{sec:introduction}

Persistent homology is an important tool of topological data analysis (TDA).
A goal of TDA is to summarize and learn from the ``shape of data''.
Often this ``shape'' is interpreted as the topological structure, such as the number of connected components and other homological features such as holes and voids.
However, persistent homology is also sensitive to geometry.
%It is the latter point that is our object of study here.

The result of a persistent homology computation may be summarized as a set of intervals called a \emph{bar code} or a set of points $(x,y)$ with $x <y$ called a \emph{persistence diagram}. These give the parameter values for which a homological feature persists. In either case, one hopes to use this summary to make inferences on the underlying object from which the data has been sampled. An oft-repeated philosophy is that the long intervals in the bar code or the points distant to the diagonal in the persistence diagram represent the ``topological signal'' while the short intervals or the points close to the diagonal represent ``noise''.
 
However, TDA has been used to understand geometric structures in many applications, such as:
force networks in particulate systems~\cite{Kramar:2014,Kondic:2016};
protein compressibility~\cite{Gameiro:2015b}; 
fullerene molecules~\cite{Xia:2015};
amorphous solids~\cite{Hiraoka:2016};
the dynamics of flow patterns~\cite{Kramar:2016};
phase transitions~\cite{Donato:2016};
sphere packing and colloids~\cite{Robins:2016};
brain arteries~\cite{bendich:brain-artery};
craze formation in glassy polymers~\cite{Ichinomiya:2017}; 
\new{branching neuronal morphologies~\cite{Kanari2017};}
and
pores in rocks~\cite{Jiang:2018}.
\new{
In these examples, the relevant geometry is the local embedding of the underlying object or the local spatial arrangement of the analyzed object.}

\new{
Here we will consider the curvature of the underlying object.
We will prove that the short intervals in the bar code can be used to infer the curvature of the underlying object that has been sampled.}
Furthermore, we will present a general framework for solving inverse problems using a continuous mapping of bar codes or persistence diagrams to a Hilbert space, called the average persistence landscape~\cite{Bub,Chazal:2015b}. We will apply this framework to learning curvature.

\subsection{Theoretical results: short bars detect geometry}
\label{sec:short-bars}

Let $D_K$ denote the unit disk in the surface of constant curvature $K$, with $K \in [-2,2]$.
For $K=0$, $K=1$, and $K=-1$, these surfaces are the Euclidean plane, the unit sphere, and the hyperbolic plane. All of these disks are contractible, so their reduced singular homology is trivial, and thus homology is unable to distinguish between them. In fact, the spaces are homeomorphic.
Endow $D_K$ with the probability measure proportional to the surface area measure. 
We will show that the persistent homology of points sampled from $D_K$ can both recover $K$ in theory and effectively estimate $K$ in practice.

We prove that for three points sampled from $D_K$ the persistence of the corresponding cycle in the \emph{\v Cech complex} is largest when the points are pairwise equidistant (Theorem~\ref{thm:isoperimetric}).
\new{
  Furthermore if this pairwise distance is fixed then we derive an analytic expression for the corresponding persistence (Theorem~\ref{thm:persistence}), which is continuous and increasing as a function of the curvature $K$ (Corollary~\ref{cor:increasing}).}
% Therefore this function is invertible and the curvature $K$ may be determined from this persistence. 
Combining these results, we have the following.

\begin{theorem} \label{thm:invertible}
  Let $p(K)$ denote the maximum (\v Cech) persistence for three points on a surface of constant curvature $K$ with pairwise distances at most some fixed constant. Then $p(K)$ is an invertible function.
\end{theorem}

We will also give several procedures for estimating $K$ from the persistent homology of the \emph{Vietoris-Rips complex} on points sampled from $D_K$. 
Before we summarize our computational results
% in Section~\ref{sec:computations}
we describe our general framework.
%in the next section.

\subsection{A framework for solving inverse problems: inference using average persistence landscapes}
\label{sec:inference}

Consider a compact metric space $(\X,d)$ together with a Borel probability measure $\mu$ with full support. Call $(\X,d,\mu)$ a \emph{metric measure space}.
Let $T$ be the diameter of $\X$.
Let $m \in \N$. 
Sample $X = (x_1,\ldots,x_m) \in \X$ independently according to $\mu$ and consider the pairwise distances $\{d(x_i,x_j) \ | \ 1 \leq i \leq j \leq m\}$.
From this data one may compute the persistent homology of the corresponding  Vietoris-Rips complex, which may be represented by the corresponding \emph{persistence landscape} $\lambda_X$~\cite{Bub}. 
%For the following, see~\cite{Chazal:2015b} for more details.
Sampling $X$ is equivalent to sampling a point in $\X^m$ according to $\mu^{\tensor m}$~\cite{Chazal:2015b}. 
Let $\Psi_{\mu}^m$ be the measure induced by $\mu^{\tensor m}$ on $\mathcal{L}$, the convex hull of the persistence landscapes of persistence diagrams consisting of at most $m$ points $(x,y)$ with $0 \leq x < y \leq T$.
The \emph{average persistence landscape} is 
$\E_{\Psi_{\mu}^m}[\lambda_X]$, the expectation of the random variable $\lambda_X$ with respect to the probability measure $\Psi_{\mu}^m$.

We may estimate the average persistence landscape as follows.
If we sample $X = (x_1,\ldots, x_m)$ as above $n$ times and average the resulting persistence landscapes, we obtain the \emph{empirical average persistence landscape}
$\bar{\lambda}^m_n = \frac{1}{n} \sum_{i=1}^n \lambda_{X^{(i)}}$.
The empirical average persistence landscape converges to the average persistence landscape (pointwise~\cite{Bub} and uniformly~\cite{Chazal:2015c}).

Now assume that $C \subset \R^d$ is a compact subset and that we have a continuous map $\varphi$ from $C$ to metric measure spaces with the Gromov-Wasserstein metric~\cite{Memoli:2011b}. 
Fix $m \in \N$. 
By~\cite[Remark 6]{Chazal:2015b}, the map from metric measure spaces with the Gromov-Wasserstein metric to their average persistence landscapes is continuous. 
Thus, composing $\varphi$ with the average persistence landscape we have a continuous map from $C$ to
%the space of persistence landscapes.
$L^2(\N \times \R)$, a Hilbert space containing the persistence landscapes and average persistence landscapes~\cite{Bub}.

Assume that for some unknown $c \in C$, we are able to sample points from the metric measure space $\phi(c)$ and compute their pairwise distances. 
In this case we can compute the empirical average persistence landscape $\bar{\lambda}_n^m(c)$. 
We now have the following inverse problem.
Given training data $\{c_i,\bar{\lambda}_n^m(c_i)\}$, can we estimate $c$ from $\bar{\lambda}_n^m(c)$?

We will demonstrate the feasibility of solving this inverse problem for the case in which $K \in [-2, 2] \subset \R$ and $\varphi(K)$ is the unit disk in the surface of constant curvature $K$ with probability measure proportional to the surface area measure.
In this case, the composition of $\varphi$ with the average persistence landscape is a parametrized path in $L^2(\N \times \R)$.
Our goal is to learn this parametrized path and to use it to estimate curvatures from empirical average persistence landscapes.

\begin{remark}
\new{
  It would be great to have an analytic derivation of the average persistence landscape for the Vietoris-Rips complex for $m$ points sampled from the unit disk in a surface of constant curvature. 
Unfortunately, not much is known in this direction.
The expected persistence diagram for the Vietoris-Rips complex for $m$ points sampled from the circle is known~\cite{bubenikKim:parametric}. 
In addition, the order of the maximally persistent degree-$k$ cycle for the Vietoris-Rips complex for $m$ points sampled from the $d$-dimensional cube
as $m \to \infty$ is known~\cite{Bobrowski:2017}.}

\new{
There is also a Vietoris-Rips complex for the unit disk in a surface of constant curvature~\cite{cdso:geometric}. This is a simplicial complex with uncountably many $k$-simplices for all $k\geq 0$. The persistent homology of the Vietoris-Rips complex for the circle has been derived analytically~\cite{Adamaszek:2017}. Note that the persistence landscape of such Vietoris-Rips complexes is not the same as the average persistence landscape for samples of $m$ points.}
\end{remark}

\subsection{Computational results}
\label{sec:computations}

We apply the framework in the previous section to estimating curvature from sampled points and pairwise distance data.

We estimate curvature in the supervised and unsupervised settings.
In the supervised setting we start with training data given by curvatures $K = \{-2,-1.96,-1.92,\ldots,1.96,2\}$
and corresponding empirical average persistence landscapes for homology in degree $0$ and homology in degree $1$, for $m=1000$ points.
In both settings, 
we sample $100$ values of $K$ iid from $[-2,2]$ and compute the corresponding empirical average persistence landscapes.
%for homology in degree $0$ and homology in degree $1$.
Using these empirical average persistence landscapes, we estimate the corresponding curvatures:
using both nearest neighbors and support vector regression in the supervised setting; and
using principal components analysis
in the unsupervised setting.
See Figure~\ref{fig:scatterplot-distance}, where we use the concatenations of the degree $0$ and degree $1$ persistence landscapes.
The root mean squared error in our estimates is
0.056
for nearest neighbors,
0.017
for support vector regression, and
0.128
for principal components analysis.
For more computational results, see Table~\ref{tab:rmse-distance}.
Furthermore, we estimate the fifth and ninety-fifth percentiles using quantile support vector regression. See Figure~\ref{fig:quantile}.

\begin{figure}[h]
\centering
\includegraphics[width=.3\linewidth]{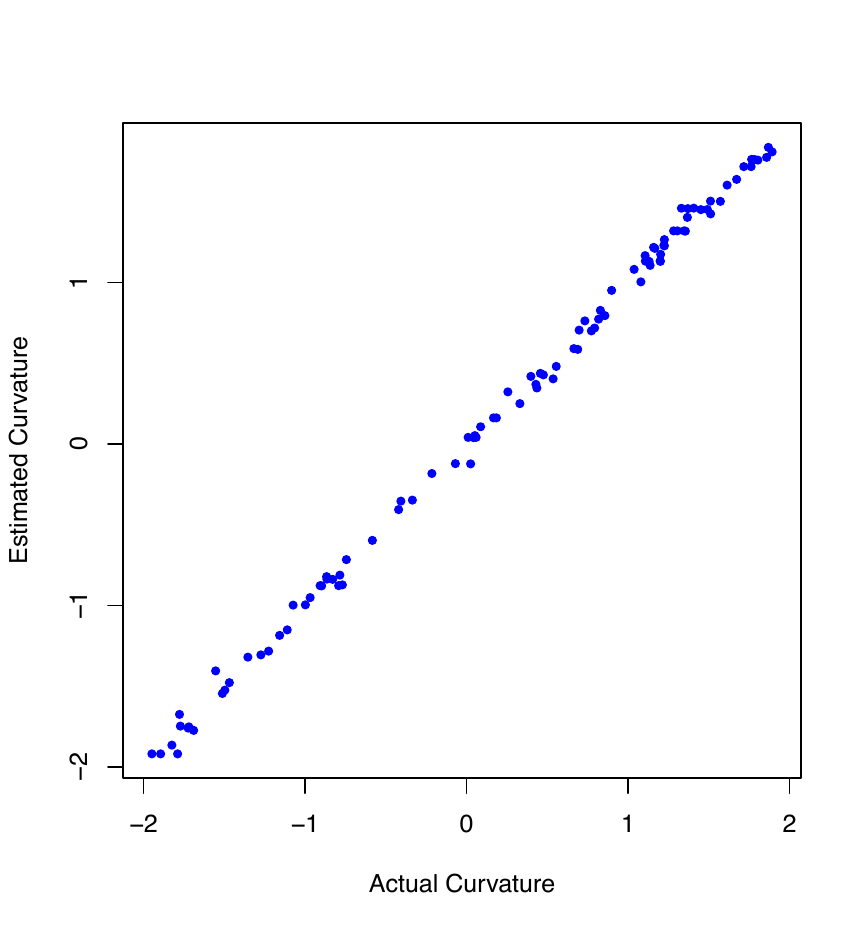}
\includegraphics[width=.3\linewidth]{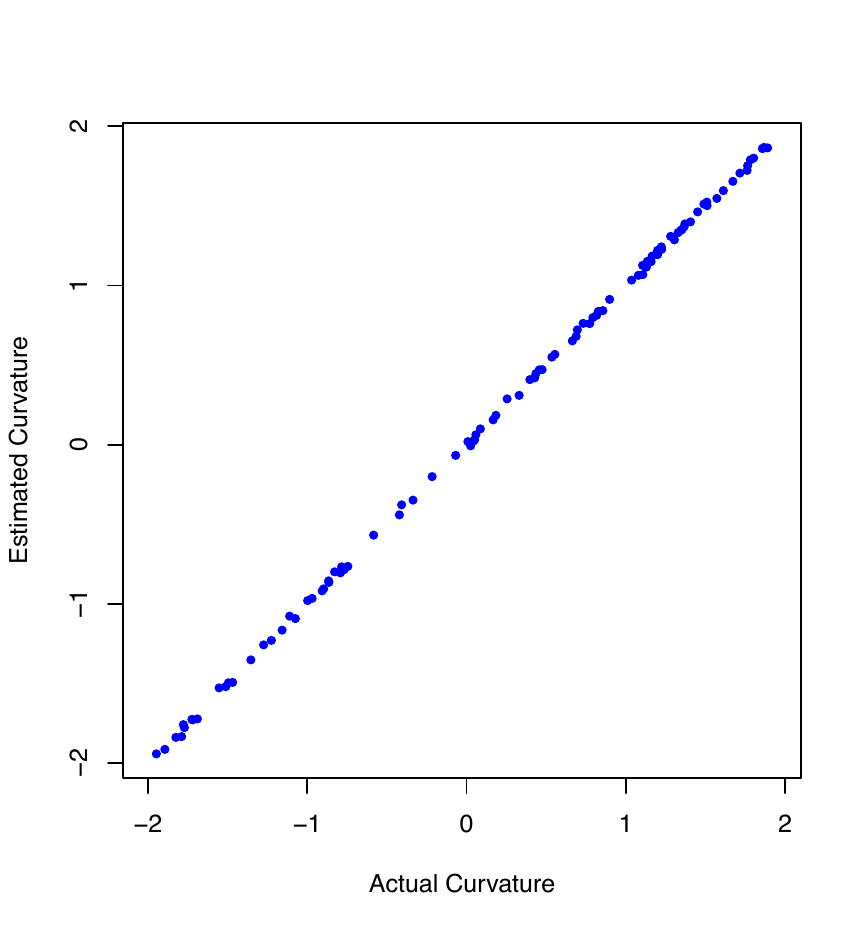}
\includegraphics[width=.3\linewidth]{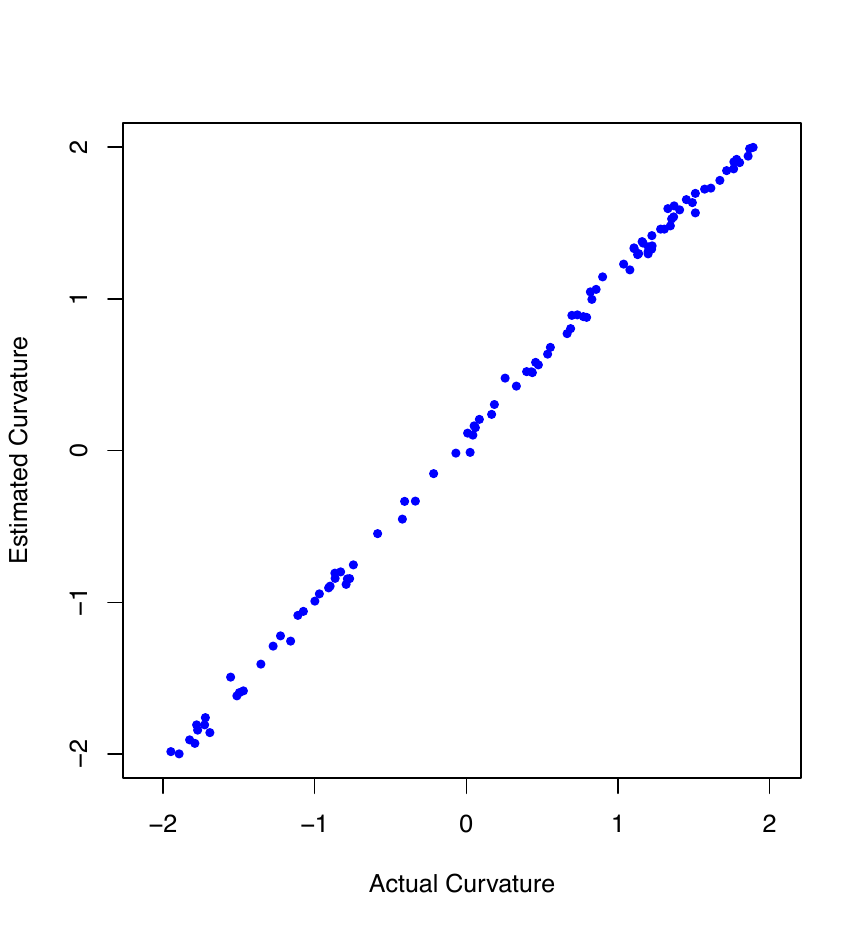}
\caption{Plots showing actual curvature and estimated curvature using  $H_0$ and $H_1$ from distance data, for nearest neighbors (left), support vector regression (center), and the first principal component (right).}
\label{fig:scatterplot-distance}
\end{figure}

We also repeat most of the above estimates for the much more difficult computational setting in which all nonzero pairwise distances are sorted and replaced with their corresponding ordinal numbers.
This is appropriate for neuroscience data in which the distances are only known up to rescaling by an unknown monotonic function~\cite{Giusti:2015}.
In this case, the set of nonzero pairwise distances is the same for all curvatures. Nevertheless, we are still able to provide reasonable curvature estimates.
See Figure~\ref{fig:scatterplot-rank}.
The root mean squared error in our estimates is
0.262
for nearest neighbors,
0.171
for support vector regression, and
0.392
for principal components analysis.
For more computational results, see Table~\ref{tab:rmse-rank}.
This example makes it clear that the short bars in persistent homology do indeed encode subtle geometric information.

\subsection{Expected impact}
\label{sec:expected-impact}

Our theoretical work
\new{
showing that persistent homology detects curvature
may be used to help justify the use of persistent homology to study
other geometric structures in applications,}
such as 
% ~\cite{bendich:brain-artery,Xia:2015,Hiraoka:2016,Jiang:2018}.
those listed in the start of the introduction.

We have outlined a framework for using topological data analysis for solving inverse problems.
Persistent homology together with the average persistence landscape gives a continuous mapping from metric spaces with a probability measure to a Hilbert space.
In situations in which it is easy to sample or subsample points and measure pairwise distances one may compute empirical average persistence landscapes.
Convergence results are known~\cite{Chazal:2015b} and in practice, they quickly converge with little noise.
Furthermore this mapping is sensitive to the starting metric structure.
Finally, as our constructions lie a Hilbert space, one can apply tools from statistical and machine learning.
This approach should facilitate learning geometric structures in a broad range of applications.

\subsection{Related work}

Persistence landscapes have been used to study the geometry of
microstructures~\cite{Dlotko:2016};
protein conformations~\cite{giseon:maltose}; 
and financial times series~\cite{Gidea:2018}.
Average persistence landscapes and average death vectors were used to detect differences in images of leaves in~\cite{Patrangenaru:2018}.
B.\ Schweinhart recently proved that persistent homology of random samples may be used to determine the fractal dimension of certain metric spaces~\cite{Schweinhart:2018a}.

\section{Background} \label{sec:background}

In this section we provide some necessary background from persistent homology, geometry, and statistics. 
For details, we refer the reader to
\cite{Edelsbrunner:2014,edelsbrunnerHarer:book,oudot:book} for persistent homology,
\cite{Coxeter,Beardon:1995,Chavel:2006} for geometry, and 
\cite{Steinwart:2008,Smola:2004} for statistics.

\subsection{Filtered simplicial complexes from points} \label{sec:sc}

A \emph{simplicial complex} is a collection $K$ of subsets of a set $V$ of vertices, such that if $\sigma \in K$ and $\tau \subset \sigma$ then $\tau \in X$.
A \emph{filtered simplicial complex} is a collection of simplicial complexes $\{K_t\;|\; t\in \mathbb R, t\geq 0\}$ with the property that whenever $s\leq t$, there is an inclusion $K_s\subseteq K_t$.

Let $Y$ be a metric space and let $X \subset Y$ be a finite subset.
There are two common ways ways to turn $X$ into a filtered simplicial complex and we will use of both of them. First, for $t\geq 0$ let $\check{C}_t(X)$ be the simplicial complex where the 0-simplices of $\check{C}_t(X)$ are the points of $X$ and for $p \geq 1$, $\check{C}_t(X)$ contains a $p$--simplex $[x_0,..., x_p]$ if and only if
\[
  \bigcap_{i=0}^{p} B_t(x_i)\neq\emptyset,
\]
where $B_r(x) \subset Y$ denotes the closed ball of radius $r$ centered at the point $x\in X$. 
The collection $\{\check{C}_t(X) : t \geq 0\}$ forms a filtered simplicial complex, called the \emph{\v Cech complex of $X$}.

Now for $t\geq 0$, let $R_t(X)$ be the simplicial complex whose 0-simplices are the points of $X$ and which includes the $p$-simplex $[x_0,...,x_p]$ if and only if for all $1\leq i, j\leq p$, $d(x_i,x_j)\leq t$.
This filtered simplicial complex is called the \emph{Vietoris-Rips Complex}.
Notice that unlike the \v Cech complex, which depends on $Y$, the Vietoris-Rips complex depends only on $X$.

\subsection{Persistent homology} \label{sec:ph}

Let $K$ be a simplicial complex. Taking reduced simplicial homology in degree $d$ with coefficients in some fixed field yields a vector space $H_d(K)$.
Furthermore an inclusion of simplicial complexes induces a linear map between the corresponding vector spaces~\cite[Chapter 8]{Armstrong:1983}.
Let $\{K_t\}$ be a filtered simplicial complex.
Taking homology in degree $d$ with coefficients in some fixed field yields a \emph{persistence module}, $M$, given by the collection of vector spaces $\{H_d(K_t)\;|\; t\in\mathbb R, t\geq 0\}$ and linear maps $f_s^t\colon H_d(K_s) \rightarrow H_d(K_t) $ induced by the inclusions $K_s \subseteq K_t$ whenever $s \leq t$.
As a special case, one has the \emph{interval persistence modules} which are one dimensional on an interval, zero outside the interval, and all linear maps are the identity whenever not forced to be zero.
%It is customary to denote the interval module on the interval $I$ by $I$ also.
The structure theorem of persistent homology says that under mild hypotheses, every persistence module $M$ is isomorphic to a direct sum of interval modules. The collection of these intervals is called the \emph{bar code} of $M$. Replacing an interval with its ordered pair of endpoints, we instead obtain the \emph{persistence diagram} of $M$.
To enable us to use ideas from statistics and machine learning, we construct the following vector summaries.

For homology in degree $0$ of both the \v Cech complex and the Vietoris-Rips complex, all of the intervals in the bar code have left endpoint $0$.
In this case we can represent the bar code by a sorted list of the right end points in decreasing order. We call this order statistic a \emph{death vector}.
Note that since we are using reduced homology and all of our complexes are eventually connected, all of the values in the death vector are finite.

In other cases, we need a more sophisticated vector encoding.
The \textit{persistent Betti number} of M corresponding to $s\leq t$ is defined to be $\beta_s^t =$ dim(image($f^t_s$)). The \textit{persistence landscape} of $M$~\cite{Bub} is the function
\[
  \lambda:\mathbb{N} \times \mathbb{R} \rightarrow \mathbb{R} : (k,t) \mapsto \sup\{ m\geq 0 : \beta_{t-m}^{t+m} \geq k\}.
\]
We discretize this function to obtain a vector,
\[
  (\lambda(1,a),\lambda(1,a+\delta),\ldots,\lambda(\new{1},a+m\delta),\lambda(2,a),\lambda(2,a+\delta),\ldots,\lambda(N,a+m\delta)),
\]
which we also call the persistence landscape.
% We can think of the persistence landscape as a decreasing sequence of non-negative functions which encode the information we need from the persistence module is a way that is convenient for statistics and machine learning.
The persistence landscape can be efficiently computed from the bar code~\cite{bubenikDlotko}.
Note that since we are using reduced homology and all of our simplicial complexes are eventually contractible, all of the values in the persistence landscape are finite.

For homology in degree $0$, we prefer the death vector to the persistence landscape since it provides a sparser encoding of the same information.

\subsection{Geometries of constant curvature}

Let $M_K$ be the complete, simply-connected 2-dimensional Riemannian manifold of constant Gaussian curvature $K$. Note that $M_K$ is unique up to isometry by the Killing-Hopf Theorem. When $K=0$, we can identify $M_0$ with $\R^2$ with the standard Euclidean metric. When $K>0$ we can identify $M_K$ with the sphere of radius $R:=\frac1{\sqrt{K}}$ centered at the origin in $\R^3$, that is $M_K=\{(x, y, z)\in\R^3\;|\; x^2+y^2+z^2=R^2\}$. When $K<0$, we identify $M_K$ with the  Poincar\'e disk model of the hyperbolic plane of curvature $K$. That is, for $R=\frac{1}{\sqrt{-K}}$,  $M_K=\{(x, y)\in\R^2\;|\; x^2+y^2<R\}$ with Riemannian metric $$ds^2=\frac{4(dx^2+dy^2)}{(1-\frac{x^2+y^2}{R^2})^{2}}.$$ The geodesics in this model correspond to the intersection of $M_K$ and a (Euclidean) line through the origin in $\R^2$ or a (Euclidean) circle which is orthogonal to the boundary circle $\{(x, y)\in\R^2\;|\; x^2+y^2=R\}$.

We think of $M_K$ as a model for hyperbolic, Euclidean, and spherical geometry when $K<0$, $K=0$, and $K>0$ respectively. The results in Section 3 will be derived using only elementary properties of these geometries. We review some of these properties next. First, however, we note that if $S$ is a surface with a Riemannian metric of constant Gaussian curvature $K$, then we can naturally identify the universal cover $\widetilde{S}$ with $M_K$. Hence $S$ will be locally isometric \new{to} $M_K$. So while the model spaces $M_K$ that we work with are all simply-connected, we will see the same behavior locally on any surface of constant curvature. Note also that by the Uniformization Theorem, every orientable surface admits a Riemannian metric of constant Gaussian curvature.

\subsection{Triangles} \label{sec:triangles}

Let $P,Q$ be distinct points in $M_K$.
Unless $K>0$ and $P$ and $Q$ are antipodal,
there is a unique line $\overleftrightarrow{PQ}$ containing $P$ and $Q$ and a unique shortest geodesic between $P$ and $Q$ whose image $\overline{PQ}$ is a subset of $\overleftrightarrow{PQ}$.

Let $A$, $B$, and $C$ be three points in $M_K$ which are assumed to not be collinear.
If $K>0$, then this implies that no pair of these points is a pair of antipodal points on the sphere. It follows that there is a unique shortest geodesic segment between each pair of points.
Let $T = \overline{AB} \cup \overline{AC} \cup \overline{BC}$ called the \emph{triangle} with \emph{vertices} $A$, $B$, $C$, and \emph{edges} or \emph{sides} $\overline{AB}$, $\overline{AC}$, $\overline{BC}$.
The subspace $M_K \setminus T$ has two components.
If $K \leq 0$ then exactly one of these has finite area, called the \emph{interior} of $T$.
If $K > 0$ then the component with smaller area is called the \emph{interior} of $T$.

\subsection{Circumcircles} \label{sec:circumcircles}

% A \emph{circle} is the set of points equidistant to some point $Q \in M_K$ called a \emph{center} of the circle. The distance from the center to the points of the circle is called the corresponding \emph{radius}. \mh{I suggest cutting the definition of a circle.}
A \emph{circumcircle} of a triangle $T$ is a circle containing the vertices of $T$. A center of this circle is called a \emph{circumcenter} and the corresponding radius is a called a \emph{circumradius}. In $M_0$, every triangle has a unique circumcircle with a unique circumcenter. If $K>0$, then each triangle in $M_K$ has a unique circumcircle with two circumcenters. If $K<0$, then a triangle in $M_K$ may or may not have a circumcircle, but if it does then the circumcenter is unique.

\begin{lemma}
Let $P$ and $Q$ be points in $M_K$. Then the perpendicular bisector of a line segment $\overline{PQ}$ consists of those points equidistant to $P$ and $Q$.
\end{lemma}

\begin{proof}
%\mh{This follows easily from Side-Angle-Side congruence condition for triangles. %which holds for hyperbolic, Euclidean, and spherical geometries.
% Or, here is a detailed proof}

Suppose $A$ is equidistant from $P$ and $Q$. Let $l$ be the line through $A$ which bisects the angle $\angle PAQ$, and let $D$ be the point where $l$ intersects $\overline{PQ}$. Then $\triangle PAD\cong \triangle QAD$ by Side-Angle-Side. Hence $\overline{PD}\cong\overline{DQ}$, so $D$ is the midpoint of $\overline{PQ}$. Also $\angle PDA\cong \angle QDA$, and since these angles sum to $\pi$ they must both be right angles. Hence $l$ is the perpendicular bisector of $\overline{PQ}$. 

Conversely, if $A$ lies on the perpendicular bisector $l$ of $\overline{PQ}$ and $D$ is the midpoint of $\overline{PQ}$, then triangles $\triangle PDA$ and $\triangle QDA$ are congruent by Side-Angle-Side, so $\overline{PA}\cong \overline{QA}$.
\end{proof}

\begin{theorem}
  For a triangle in $M_K$, the following statements are equivalent.
  \begin{enumerate}
  \item The perpendicular bisectors of two of the sides intersect.
  \item The triangle has a circumcircle.
  \item The perpendicular bisectors of the sides have a common intersection.
  \end{enumerate}
Moreover, when \new{at least one of} these equivalent statements holds then the intersection point of the perpendicular bisectors of the sides is the circumcenter of the triangle.
\end{theorem}

\begin{proof}
  Let $A$, $B$, $C$ be the vertices of triangle $T$.\\
  (a) implies (b). Assume that a point $P$ is in the intersection of the perpendicular bisectors of two of the sides of $T$. Then $P$ is equidistant from $A$, $B$, and $C$. So $P$ is a circumcenter of $T$.

  (b) implies (c). Let $P$ be a circumcenter. Then $P$ is equidistant from $A$, $B$, $C$. So $P$ lies on the perpendicular bisector of each side.

  (c) implies (a) is immediate.
\end{proof}

\subsection{Areas of disks}

We will use the following basic fact.
% For $K=0$, in Euclidean geometry, the area of a circle of radius $r$ is $\pi r^2$.
% For $K>0$, in spherical geometry, let $R = \frac{1}{\sqrt{K}}$.
% The area of a circle of radius $r$ is $4\pi R^2 \sin^2(\frac{r}{2R})$.
% For $K<0$, in hyperbolic geometry, let $R = \frac{1}{\sqrt{-K}}$.
% The area of a circle of radius $r$ is $4\pi R^2 \sinh^2(\frac{r}{2R})$.
The area of a disk of radius $r$ on a surface of constant curvature $K$ is given by
\begin{equation*}
  A(r) =
  \begin{cases}
    \frac{4\pi}{-K} \sinh^2 \left( \frac{r\sqrt{-K}}{2}\right) & \text{if } K<0\\
    \pi r^2 & \text{if } K=0\\
    \frac{4\pi}{K} \sin^2 \left( \frac{r\sqrt{K}}{2}\right) & \text{if } K>0\new{.}
  \end{cases}
  % \begin{cases}
  %   \pi r^2 & \text{if } K=0\\
  %   \pi \left[ \frac{2}{\sqrt{K}} \sin \left( \frac{r\sqrt{K}}{2}\right) \right]^2 & \text{if } K>0\\
  %   \pi \left[ \frac{2}{\sqrt{-K}} \sinh \left( \frac{r\sqrt{-K}}{2}\right) \right]^2 & \text{if } K<0
  % \end{cases}
\end{equation*}

\subsection{Distances between points on  a unit disk}
\label{sec:distance}

We will want to compute the distances between points sampled from a disk of radius one on $M_K$. We will represent the points in this disk using polar coordinates $(r,\theta)$, where $0 \leq r \leq 1$ and $0 \leq \theta < 2\pi$.

For the Euclidean case, $K=0$, we convert to Cartesian coordinates $(r\cos\theta, r\sin\theta)$ and compute the Euclidean distance.

In the spherical case, $K>0$, $M_K$ is realized as the sphere of radius $R$ centered at the origin in $\R^3$, where $R = \frac{1}{\sqrt{K}}$. We consider our disk to be a spherical cap of this sphere. The point on the disk corresponding to $(r, \theta)$ can be written in spherical coordinates as $(R,\theta,\frac{r}{R})$.
Converting to Cartesian coordinates, we have
$(R\sin(\frac{r}{R})\cos\theta,R\sin(\frac{r}{R})\sin\theta,R\cos(\frac{r}{R}))$.
The distance between two such points $x$ and $y$ is given by
$R \cos^{-1}(\frac{x \cdot y}{R^2})$.
However, $\cos^{-1}(t)$ is not numerically stable near zero, so instead we use the following robust formula,
$R \tan^{-1}(\frac{|x \times y|}{x \cdot y})$. More specifically, we will use the two-argument arctangent function $R \atan\!2(|x \times y|,x \cdot y)$.

For the hyperbolic case, $K<0$, $M_K$ is realized as the Poincar\'e disk, with $R = \frac{1}{\sqrt{-K}}$.
We consider our disk of hyperbolic radius one to be centered at the origin. The point on the disk corresponding to $(r, \theta)$ can be written in Cartesian coordinates as
$\left(R\tanh\left(\frac{r}{2R}\right)\cos\theta,R\tanh\left(\frac{r}{2R}\right)\sin\theta\right)$.
The hyperbolic distance between between to points $u$ and $v$ in the Poincar\'e $R$-disk is given by $2R \tanh^{-1} \frac{|z-w|}{|1-z\bar{w}|}$ where $z = u/R$ and $w = v/R$ are thought of as complex numbers.

\subsection{Laws of sines and cosines}

We will need the laws of sines and cosines for a triangle on a surface of constant curvature $K$.

\begin{theorem}{Generalized Law of Sines}%~\cite{Coxeter}.
  \label{thm:sine-law}

Let $\Delta ABC$ be a triangle in $M_K$ with lengths $a,b,c$ and angles $\alpha,\beta,\gamma$ respectively. 
When K=0, 
$$\dfrac{sin(\alpha)}{a} = \dfrac{sin(\beta)}{b} = \dfrac{sin(\gamma)}{c}.$$
When $K>0$,
$$\dfrac{\sqrt{K}sin(\alpha)}{sin(a\sqrt{K})} = \dfrac{\sqrt{K}sin(\beta)}{sin(b\sqrt{K})} = \dfrac{\sqrt{K}sin(\gamma)}{sin(c\sqrt{K})}.$$
When $K<0$,
$$\dfrac{\sqrt{-K}sin(\alpha)}{sinh(a\sqrt{-K})} = \dfrac{\sqrt{-K}sin(\beta)}{sinh(b\sqrt{-K})} = \dfrac{\sqrt{-K}sin(\gamma)}{sinh(c\sqrt{-K})}.$$
\end{theorem}

 \begin{theorem}{Generalized Law of Cosines}\label{thm:cosines-law}

Let $\Delta ABC$ be a triangle in $M_K$ with lengths $a,b,c$ and angles $\alpha,\beta,\gamma$ respectively. 
When $K=0$, 
\[
c^2=a^2+b^2+ab\cos(\gamma)
\]
When $K>0$, 
\[
\cos(c\sqrt{K})=\cos(a\sqrt{K})\cos(b\sqrt{K})+\sin(a\sqrt{K})\sin(b\sqrt{K})\cos(\gamma)
\]
When $K<0$,
\[
\cosh(c\sqrt{-K})=\cosh(a\sqrt{-K})\cosh(b\sqrt{-K})-\sinh(a\sqrt{-K})\sinh(b\sqrt{-K})\cos(\gamma).
\] 
\end{theorem}

\subsection{Inversion sampling} \label{sec:inversion-sampling}
  
The following theorem allows us to sample points from a distribution knowing only the inverse of the cumulative distribution $\new{F}$ by sampling a point $u$ uniformly from $[0,1]$ and then calculating $F^{-1}(u)$. This method of sampling from F is called \emph{inversion sampling}. 

\begin{theorem}\cite{Devroye} \label{thm:inversion-sampling}
  Let F be an invertible continuous cumulative distribution function on some domain D. If U is distributed uniformly on $[0,1]$, then $F^{-1}(U)$ has cumulative distribution function F. 
\end{theorem}

\subsection{Support vector regression} \label{sec:svr}

We assume that our data $(x_1,y_1),\ldots,(x_N,y_N)$ with $x_i \in \R^d$ and $y_i \in \R$ is drawn from some unknown joint distribution on $\R^d \times \R$.
Our goal is to estimate a functional relationship between the variables.

\emph{(Linear) support vector regression} (SVR) is an approach to this problem which computes a predictor $f(x) = \langle w, x \rangle + b$ by solving the following convex optimization problem:
\begin{align*}
\text{minimize}\quad &
\dfrac{1}{2} \|w\|^2 + C\sum\limits_{i=1}^N (\zeta_{1,i}+\zeta_{2,i})\\
\text{subject to}\quad &
  \begin{cases}
     y_i - (\langle w,x_i \rangle +b) \leq \new{\varepsilon} + \zeta_{1,i}\\
     (\langle w,x_i \rangle +b) - y_i \leq \new{\varepsilon} + \zeta_{2,i}\\
     \zeta_{1,i},\zeta_{2,i} \geq 0
   \end{cases}
\end{align*}
for each $i\in \{1,...,N\}$.
The slack variables $\zeta_{1,i}$ and $\zeta_{2,i}$ and cost parameter $C$ allow for some errors among the training data.
A larger value of $C$ increases the penalty for an error in the training data.
If $\new{\varepsilon}=0$ then this problem corresponds to using the \emph{linear loss function} $L = |y_i-f(x_i)|$. 
For $\varepsilon>0$, we instead have the \emph{$\varepsilon$-insensitive loss function} given by
\[L_{\varepsilon-\ins} = \begin{cases} 
      0 & \text{if } |y_i-f(x_i)| \leq \varepsilon \\
      |y_i-f(x_i)| - \varepsilon & \text{otherwise.} \\ 
    \end{cases}
    \]
This function ignores errors within $\varepsilon$ of the true values.

In Section~\ref{sec:svr-comp}, in which the data seems to have little noise, we are able to set $\varepsilon=0$ and $C=100$.
In Section~\ref{sec:rank}, in which the data seems to be noisier,
we take $\varepsilon = 1$ or $0.2$ and $C=10$
to avoid overfitting.

For \emph{quantile regression} (Section~\ref{sec:quantile}) we will use the \emph{pinball loss function},
\begin{equation*}
  L_{\tau-\pin} =
  \begin{cases}
    (\tau-1)(y_i - f(x_i)) & \text{if } y_i < f(x_i)\\
    \tau (y_i - f(x_i)) & \text{if } y_i \geq f(x_i),
  \end{cases}
\end{equation*}
where $0 < \tau < 1$.
This loss function allows us to estimate the $\tau$-quantile. 

\section{Persistence of triangles}

In this section we study how triangles contribute to the persistent homology of the \v{C}ech complex formed from points on $M_K$. Specifically, we show the maximal interval of parameter values for which three points contribute a non-trivial element to the homology in degree one of the \v{C}ech complex depends on $K$. Moreover, we will show that for all $K$ this interval is maximized by the vertices of an equilateral triangle $T$ and give formulas depending on $K$ and the length of the sides of $T$.

\subsection{Triangles and their persistent homology}

Let $X$ be a finite set of points on $M_K$.
Let $A$, $B$ and $C$ be points in $X$ which we assume are not collinear. When $K>0$, we will also assume that no pair of these points is antipodal, or equivalently the pairwise distances are all less then $\frac{\pi}{\sqrt{K}}$.
There are two triangles of interest corresponding to the vertices $A$, $B$, and $C$.
There is the (geometric) triangle $T$, which is a subset of $M_K$ (Section~\ref{sec:triangles}).
There is also the abstract triangle $\{A,B,C\} \subset X$ which may be an element of the \v Cech complex on $X$.
It will be convenient to refer to both of these as the triangle $T$ corresponding to the vertices $A$, $B$, and $C$.
It should be clear from the context which of these we mean.

%\pb{Since we are intersecting closed balls, I replaced inf with min in the two expressions below. This change is helpful in some of the proofs below.} \mh{Great!}

The boundary of the triangle $T$ contributes 
a 1--cycle in $\check{C}_t(X)$ in the \v{C}ech  complex whenever $t\geq b(T)$, where 
\[
b(T)=\min\{r\;|\; B_r(X)\cap B_r(Y)\neq\emptyset\;\forall X, Y\in \{A, B, C\}\}.
\] 
The value $b(T)$ is called the \emph{birth} of the triangle $T$.
In the other direction, $T$ contributes a 2--simplex to $\check{C}_t(X)$ whenever $t\geq d(T)$, where
\[
d(T)=\min\{r\;|\; B_r(A)\cap B_r(B)\cap B_r(C)\neq\emptyset\}.
\]
The value of $d(T)$ is called the \emph{death} of $T$. Hence $T$ induces an element of $H_1(\check{C}_t(X))$ for all $t\geq b(T)$, and this element is trivial for all $t\geq d(T)$. In particular, if $b(T)=d(T)$, then $T$ does not contribute any non-trivial elements to the persistent homology.

The persistence of an interval $[b(T),d(T))$ is usually given by
the difference $d(T)-b(T)$, the length of the time that $T$ is %(potentially)
contributing to homology.
However, in situations where a scale-free version is desired~\cite{Bobrowski:2017}, it is preferable to use logarithmic coordinates and to instead consider the ratio $\frac{d(T)}{b(T)}$, which we will refer to as the \emph{persistence of $T$} and denote by $p(T)$.
  
%\subsection{Triangles with persistent homology}

We now we fix notation that will be used for the rest of this section. Let $T$  denote a triangle with vertices $A$, $B$, and $C$. Let $a$, $b$, and $c$ be the lengths of the sides of $T$ opposite $A$, $B$, and $C$, respectively.
We assume $T$ is labeled such that $a\geq b\geq c$. See Figure~\ref{fig:basic-triangle}. When $K>0$, we let $R=\frac{1}{\sqrt{K}}$, that is $R$ is the radius of the sphere realizing $M_K$. Recall that in this case we are also assuming that $a<\pi R$. Note that the birth of $T$ is simply half the length of the longest side, so with this notation we have $b(T)=\frac{a}{2}$. We will also let $M$ denote the midpoint of the side $\overline{BC}$, and let $m$ denote the distance from $A$ to $M$. If $T$ has a circumcircle then we denote the corresponding circumcenter by $P$.

\begin{figure}[h] \centering
  \begin{tikzpicture}[scale=1.5]
    \draw [color=black] (0,0) -- (4,0) -- (1.98,2.34) -- (0,0);
    \draw[dashed] (1.98,2.34)-- (2,0); 
    \draw[fill] (0,0) circle [radius=0.05];
    \draw[fill] (4,0) circle [radius=0.05];
    \draw[fill] (1.98,2.34) circle [radius=0.05];
    \draw[fill] (2,0) circle [radius=0.05];
    \node [left] at (0,0) {$B$};
    \node [right] at (4,0) {$C$};
    \node [right] at (1.98,2.34) {$A$};
    \node [below] at (2,0) {$M$};
    \node [below] at (1,0) {$\frac{a}{2}$};
    \node [below] at (3,0) {$\frac{a}{2}$};
    \node [above right] at (3,1) {$b$};
    \node [above left] at (1,1) {$c$};
  \end{tikzpicture}
  \caption{$\triangle ABC$.} \label{fig:basic-triangle}
\end{figure}
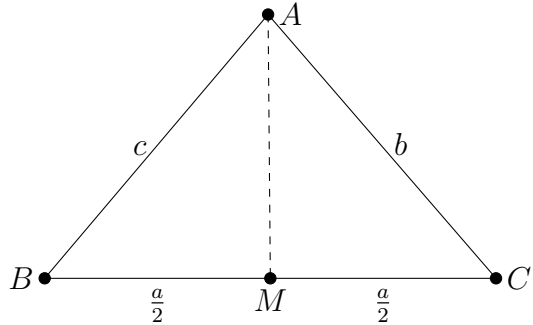

In this section we will prove the following.

\begin{proposition}\label{prop:circumcenter}
  The following are equivalent:
  \begin{enumerate}
  \item $T$ produces persistent $H_1$ in the \v Cech complex. That is, $b(T) < d(T)$.
  \item $\frac{a}{2} < m$.
  \item $T$ has a circumcircle and the circumcenter $P$ is in the interior of $T$.
  \end{enumerate}
  Furthermore, if these equivalent conditions hold, then $b(T)$ equals $\frac{a}{2}$ and $d(T)$ equals the circumradius.
\end{proposition}

\begin{lemma}\label{lem:distincreasing}
  Let $P$ and $Q$ be points in $M_K$ and let $l$ be a line through $Q$ which is perpendicular to $\overleftrightarrow{PQ}$.
  If $K>0$, then we also assume that $d(P, Q)<\frac{\pi}{2}R$.
  Let $t\geq 0$.
  Let $H$ be a half plane bounded by $\overleftrightarrow{PQ}$ and let $Q_t$ be the point in $H$ on $l$ such that $d(Q, Q_t)=t$. Then $d(P, Q_t)$ is a strictly increasing function of $t$.
\end{lemma}
\begin{proof}
Since $\cos \angle PQQ_t = 0$, this follows from the Generalized Law of Cosines (Theorem \ref{thm:cosines-law}).
\end{proof}

The next lemma explains the role of the distance $m$ from $A$ to the midpoint $M$ of $\overline{BC}$.

\begin{lemma} \label{lem:persistent-triangle}
  $b(T)<d(T)$ if and only if  $\frac{a}{2}<m$.
\end{lemma}

\begin{proof}
Note that $M$ is the unique element of $B_{\frac{a}{2}}(B)\cap B_{\frac{a}{2}}(C)$. Hence $B_{\frac{a}{2}}(A)\cap B_{\frac{a}{2}}(B)\cap B_{\frac{a}{2}}(C)\neq\emptyset$ if and only if $M\in B_{\frac{a}{2}}(A)$, thus $b(T)=d(T)$ if and only if $m\leq\frac a2$. Therefore, $b(T)<d(T)$ if and only if  $\frac{a}{2}<m$.
\end{proof}

\begin{lemma}\label{interiorCC}
  Suppose that $\frac{a}{2}<m$. Then the triangle $T$ has a circumcenter $P$, and $P$ is contained in the interior of $T$.
\end{lemma}
 
\begin{proof}
Assume that $\frac{a}{2}<m$. Note that the distance from $M$ to $C$ is $\frac{a}{2}$, so $d(M, A)>d(M, C)$. Let $l$ denote the perpendicular bisector of $\overline{BC}$. Then $l$ must intersect one of the other two sides of $T$.
% without loss of generality say
Since $b \geq c$,
$l$ intersects $\overline{AC}$ in a point $N$. If $K>0$, then $d(M, C)=\frac{a}{2}<\frac{\pi}{2}R$. Hence for any $K$ we can apply Lemma \ref{lem:distincreasing} to get that $d(N, C)>d(M, C)=\frac{a}{2}$. Since the length of $\overline{AC}$ is $b\leq a$, we must have $d(A, N) = d(A,C) - d(N,C) < a - \frac{a}{2} = \frac{a}{2}$. Thus $d(N, A)< d(N, C)$.

Now, as a point moves along $l$ from $M$ to $N$, by the continuity of the distance function and the intermediate value theorem there must exist a point $P$ in the interior of $\overline{MN}$ where the $d(P, A)=d(P, C)$. Since $P$ is on the perpendicular bisector of $\overline{BC}$, we also have the distance from $P$ to $B$ is equal to the distance from $P$ to $C$. Thus, $P$ is a circumcenter of $T$. Since $P$ is in the interior of $\overline{MN}$ and this segment is contained in $T$ by construction,  we have that $P$ is in the interior of $T$.
\end{proof}

\begin{lemma}\label{lem:d=cr}
Suppose $\frac{a}{2}<m$ and there exists an $r>0$ and a point $D$ such that $D\in B_r(A)\cap B_r(B)\cap B_r(C)$ but $D$ is not the circumcenter of triangle $\triangle ABC$. Then there exists $r'<r$ such that  $B_{r'}(A)\cap B_{r'}(B)\cap B_{r'}(C)\neq\emptyset$.  
\end{lemma}
\begin{proof}
Since $D$ is not the circumcenter, there must exist at least one vertex whose distance to $D$ is less then $r$. Suppose without loss of generality that $d(D, A)<r$. First, suppose that $D\notin \overleftrightarrow{BC}$. Let $l$ be a line that contains $D$ and is perpendicular to $\overleftrightarrow{BC}$. Let $D^\prime\neq D$ be a point on the segment of $l$ from $D$ to $\overleftrightarrow{BC}$ such that $d(D, D')<r-d(D, A)$. Hence $d(D', A)<r$. Also, by construction and Lemma~\ref{lem:distincreasing} $D'$ is closer to $B$ and $C$ than $D$, hence $d(D', B)<d(D, B)\leq r$ and similarly $d(D', C)<r$. Letting $r'=\max\{d(D', A), d(D', B), d(D', C)\}$, we get that $r'<r$ and $D'\in B_{r'}(A)\cap B_{r'}(B)\cap B_{r'}(C)$.  

Now suppose $D\in  \overleftrightarrow{BC}$. Suppose without loss of generality that $d(D, B)\leq d(D, C)$. Then either $d(D, B)<d(D, C)\leq r$, or $d(D, B)=d(D, C)=\frac{a}{2}<m\leq r$ since $D$ is the midpoint of $\overline{BC}$ in this case. Either way we get that $d(D, B)<r$, so we can repeat the same proof as above using the point $B$ instead of $A$. 
\end{proof}

\begin{proof}[Proof of Proposition \ref{prop:circumcenter}]
(a) and (b) are equivalent by Lemma \ref{lem:persistent-triangle}. Lemma \ref{interiorCC} shows that (b) implies (c). Assume now that (c) holds, that is $T$ has a circumcircle and the circumcenter $P$ is in the interior of $T$. Since $P$ lies on the perpendicular bisector $l$ of $\overline{BC}$,  by Lemma \ref{lem:distincreasing} we get that $\frac{a}{2}=d(B, M)<d(B, P)=d(A, P)$.

Now let $Q$ be a point on $l$ such that $\overleftrightarrow{AQ}$ is perpendicular to $l$. Then $\overleftrightarrow{AQ}$ and $\overleftrightarrow{BC}$ are both perpendicular to $l$. When $K\leq 0$, this means that $\overleftrightarrow{AQ}$ and $\overleftrightarrow{BC}$ are parallel. When $K>0$, this means that the two intersection points of $\overleftrightarrow{AQ}$ and $\overleftrightarrow{BC}$ both have distance $\frac{\pi}{2}R$ from $l$. Furthermore, the line $l$ must intersect either $\overline{AB}$ or $\overline{AC}$.
% Without loss of generality, suppose
Since $b \geq c$,
$l$ intersects $\overline{AC}$. So $d(A, Q)\leq d(A, C)\leq a<\frac{\pi}{2}R$. Thus, for all $K$ we get that $\overline{AQ}$ does not intersect $\overleftrightarrow{BC}$. Thus, $Q$ and $A$ are on the same side of $\overleftrightarrow{BC}$.
By a similar argument, it also follows that $\overleftrightarrow{AQ}$ does not intersect $\overline{BC}$.

We also note that $Q$ cannot be in the interior of $T$. Indeed, suppose $Q$ is in the interior of $T$, and \new{let} $S$ be the point where $l$ and $\overline{AC}$ intersect. Hence $Q$ lies in the interior of the segment $\overline{MS}$. Since $\overleftrightarrow{AQ}$ intersects one side of triangle $\triangle MSC$, it must intersect one of the other two sides. We have already shown that $\overleftrightarrow{AQ}$ does not intersect $\overline{MC}\subseteq\overline{BC}$, hence it must intersect $\overline{SC}\subseteq \overleftrightarrow{AC}$. But this means that $\overleftrightarrow{AQ}$ and $\overleftrightarrow{AC}$ must intersect in two non-antipodal points, a contradiction.

%\pb{Does orthogonal projection (below) work for all $K$?}

%We also note that $Q$ cannot be in the interior of $T$. Let $S$ be the point where $l$ and $\overline{AC}$ intersect. Consider the orthogonal projection of $\overline{AC}$ to $l$. This projects $C$ to $M$, $A$ to $Q$, and fixes the point $S$. Since orthogonal projection preserves betweenness and $S$ is between $A$ and $C$ on $\overleftrightarrow{AC}$, we get that $S$ is between $Q$ and $M$ on $l$. Hence $Q$ and $M$ are on opposite sides of $\overleftrightarrow{AC}$, which means that $Q$ is not in the interior of $T$. 

%N Geo proof: Let $S$ be the point where $l$ intersects $\overline{AB}$. Since $\triangle BSM$ has a right angle at vertex $M$, the angle at vertex $S$ must be acute. It follows that angle at $S$ of $\triangle ASM$ must be obtuse. Hence, for any point $Q'$ on $\overline{SM}$, triangle $ASQ'$ has an obtuse angle at vertex $S$ so the angle at vertex $Q'$ must be acute. Hence, $Q'\neq Q$, that is $Q$ must not be in the interior of $T$.

Since $Q$ and $P$ are on the same side of $\overleftrightarrow{BC}$ and $P$ is in the interior of $T$ and $Q$ is not, we must have $P\in\overline{QM}$. Thus, we get that $d(P, Q)<d(M, Q)$, and so by Lemma \ref{lem:distincreasing} $d(A, P)<d(A, M)=m$. Combining this with the previous inequality gives that $\frac{a}{2}<m$.

 % Finally, suppose (a), (b), and (c) hold. Let $r=d(T)$. By definition, for all $t> r$,  $B_t(A)\cap B_t(B)\cap B_t(C)\neq\emptyset$. Hence by compactness, $B_r(A)\cap B_r(B)\cap B_r(C)=\bigcap_{t>r}B_t(A)\cap B_t(B)\cap B_t(C)\neq\emptyset$. Lemma \ref{lem:d=cr} then implies that $P$ must be the unique element of $B_r(A)\cap B_r(B)\cap B_r(C)$, and hence $r=d(P, A)=d(P, B)=d(P, C)$, that is $r$ is the circumradius of $T$.

Finally, suppose (a), (b), and (c) hold. Let $r=d(T)$. By definition, $B_r(A)\cap B_r(B)\cap B_r(C)\neq\emptyset$.
Lemma \ref{lem:d=cr} then implies that $P$ must be the unique element of $B_r(A)\cap B_r(B)\cap B_r(C)$, and hence $r=d(P, A)=d(P, B)=d(P, C)$, that is $r$ is the circumradius of $T$.
\end{proof}
% Now let $T$ be a triangle for which $b(T)<d(T)$. Suppose $P$ is a point of  $B_{d(T)}(A)\cap B_{d(T)}(B)\cap B_{d(T)}(C)$. Then $P$ must be on the boundary of each of these closed balls, otherwise there would be some $\epsilon>0$ such that $B_{d(T)-\epsilon}(A)\cap B_{d(T)-\epsilon}(B)\cap B_{d(T)-\epsilon}(C)\neq \emptyset$ contradicting the definition of $d(T)$. It follows that $P$ is equidistant from $A$, $B$, and $C$, that is $P$ must be the unique circumcenter of these three points. Moreover, $d(T)$ is equal to the circumradius of the triangle $T$. Applying the formulas for the circumradius of a triangle in a space of constant curvature $K$ we obtain a formula for $d(T)$ in each case.

\subsection{The most persistent triangles}

In this section we show that among triangles $T$ with fixed birth $b(T)$, those with maximal persistence $p(T)=\frac{d(T)}{b(T)}$ are the equilateral triangles.

Let $T$ be a triangle with vertices $A$, $B$, and $C$, and corresponding edge lengths $a \geq b \geq c$.
Assume that $b(T) < d(T)$.
If $K>0$ then we also assume that 
%$a < \pi R$, 
$a < \frac{2\pi}{3}R$, 
where $R = \frac{1}{\sqrt{K}}$. This assumption is necessary for an equilateral triangle with side lengths $a$ to exist on $M_K$.

\begin{theorem} \label{thm:isoperimetric}
Suppose $T$ is not an equilateral triangle. Then there exists an equilateral triangle $T^\prime$ such that $b(T^\prime)=b(T)$ and $d(T^\prime)>d(T)$.
\end{theorem}
%(.85, .945)
%(.8, .94)
%(.82, .98)
%(.83, .99)
\begin{figure}[h]
	\centering
	\begin{tikzpicture}[scale=1.5]
									\draw [color=black] (0,0) -- (4,0) -- (1.98,2.34) -- (0,0);
									\draw [color=blue](2.38, 2.85)-- (4,0);
									\draw[fill] (0,0) circle [radius=0.05];
									\draw[fill] (4,0) circle [radius=0.05];
									\draw[fill] (1.98,2.34) circle [radius=0.05];
									\draw[fill] (2,.32) circle [radius=0.05];
									\draw[fill] (2,0) circle [radius=0.05];
									\draw[fill] (.83, .99) circle [radius=0.05];
									\draw[fill] (2.38, 2.85) circle [radius=0.05];%A'
									\node [left] at (0,0) {B};
									\node [right] at (4,0) {C};
									\node [right] at (1.98,2.34) {A};
									\node [below] at (2,0) {M};
									\node [right] at (2,.32) {P};
									
									\draw[fill] (2.37,0) circle [radius=0.05];
									\node [below] at (2.37,0) {N};
									
									\draw[dashed] (0,2.01)--(3.56,-1);
									%\draw[dashed] (3.98,2.02)--(.5,-.98);
									\draw[dashed] (2,3)-- (2,-1);
									%\draw[dashed] (2, 0)--(.83, .99);
									\draw[dashed] (0, 1.7)--(3.19,-1);
									\draw[color=black] (1.98, 2.34)--(2.38, 2.85);
									
									\draw[dashed, color=blue] (0,2.43) -- (4.02,-.97);
									\draw[fill, color=blue] (2,.74) circle [radius=.05];
									\node [right] at (2,.74) {$P'$};
									\node [left] at (0,2.01) {$l_2$};
									\node [left] at (0,2.43) {$l'_2$};
									\node [left] at (2, -1) {$l_1$};
									\node[left] at (.83, .99){Q};	
									\node[right] at (2.38, 2.83){A$^\prime$};									

									%\draw[color = green](2,.32) -- (0,0);
									%\draw[color = orange](2,.78) -- (0,0);
									\end{tikzpicture}	
\caption{Replacing $\triangle ABC$ with an isosceles triangle $\triangle A'BC$.} \label{fig:triangle}
\end{figure}
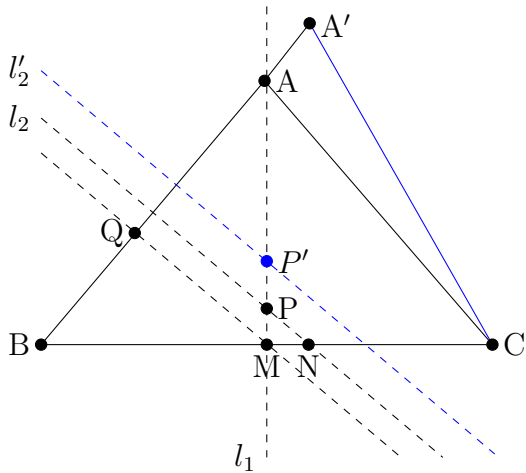
%\pb{We should add the point $N$}.\mh{done}
\begin{proof}
%This proof will hold in both neutral and spherical geometry, and hence applies for all $K$. There is only one claim that we have to be careful about when $K>0$. In neutral geometry, if $l$ is a line containing $Y$ and perpendicular to a line $\overleftrightarrow{XY}$, then the distance from $X$ to a point $P$ on $l$ is strictly increasing as $P$ moves away from $Y$. This will still be true on a sphere of radius $R$ if and only if $d(X, Y)<\frac{\pi R}{2}$. In all cases where we apply this, $Y$ will be the midpoint of a geodesic segment from $X$ to a point which is not the antipode of $X$, and hence $d(X, Y)<\frac{\pi R}{2}$ will be satisfied.--\mh{this comments have now been put into lemma \ref{lem:distincreasing}.

We will first show that  $T$ can be replaced by an isosceles triangle with two sides of length $a$. If $T$ is not already of this form, then longest side of $T$ is strictly bigger then the length of the other two sides, that is $a>b\geq c$.  Let $l_1$, $l_2$, and $l_3$ be the perpendicular bisectors to $\overline{BC}$,$\overline{AB}$,and $\overline{AC}$ respectively.
% \pb{Why is the following statement true?}\mh{added ref to Proposition \ref{prop:circumcenter}}
By Proposition \ref{prop:circumcenter},
these bisectors intersect in the point $P$, that is the circumcenter of $T$, which is in the interior of $T$.

Let $A'$ be the point on $\overleftrightarrow{AB}$ such that $A$ is between $A'$ and $B$ and $\max\{d(B, A'), d(C, A'\}=a$. Let $T'$ be the triangle formed by $A'$,$B$, and $C$.
See Figure~\ref{fig:triangle}.
By construction, $T$ has two sides of length $a$, and $a$ is still the length of the longest side of $T^\prime$.
Thus $b(T') = b(T)$.

We will show that $T'$ satisfies $d(T')>b(T')$ using Proposition \ref{prop:circumcenter}.
Let $M$ be the midpoint of $\overline{BC}$.
% We first observe that $l_2\cap \overline{BM}=\emptyset$. Is because $l_2$ intersects $l_1$ in the interior of $T$, and hence the point where $l_2$ exists the triangle $T$ must occur on the side of $l_1$ opposite $B$.
Since $a > b$, $l_2$ intersects $\overline{BC}$ at a point $N$.
Since $P$ is inside $T$, $N$ is on the opposite side of $l_1$ as $B$.

This means that $M$ and $B$ are on the same side of $l_2$, and hence $M$ and $A$ are on opposite sides of $l_2$.

 Now let $Q$ be the point on $\overleftrightarrow{AB}$ such that $\overleftrightarrow{MQ}$ is perpendicular to $\overleftrightarrow{AB}$. Then $\overleftrightarrow{MQ}$ and $l_2$ are both perpendicular to $\overleftrightarrow{AB}$. When $K\leq 0$, this means that $\overleftrightarrow{MQ}$ and $l_2$ are parallel. When $K>0$, this means that the two intersection points of $l_2$ and $\overleftrightarrow{MQ}$ both have distance $\frac{\pi}{2}R$ from $\overleftrightarrow{AB}$. In this case, $d(M, Q)\leq d(M, B)=\frac{a}{2}<\frac{\pi}{2}R$. Hence for all $K$ we get that the segment $\overline{MQ}$ does not intersect $l_2$. Thus $M$ and $Q$ are on the same side of $l_2$, which means that $Q$ and $A$ are on opposite sides of $l_2$.

 Since $A$ is closer to $l_2$ \new{than} $A'$, it follows that $A$ is closer to $Q$ \new{than} $A'$. Hence Lemma \ref{lem:distincreasing} implies that $d(A', M)>d(A, M)>\frac a2$, which means that the conclusions of Proposition \ref{prop:circumcenter} hold for $T'$.

%\pb{Why does the first inequality hold?}\mh{added explanation}

%\pb{It would be good to add a figure for this. }\mh{Agreed.}

Let $P'$ be the circumcenter of triangle $T'$ and $l'_2$ be the perpendicular bisector of $\overline{BA'}$. Then $P'$ lies on $l_1$ and $d(M, P') > d(M, P)$. Since $l_1$ is perpendicular to $\overleftrightarrow{BM}=\overleftrightarrow{BC}$, the distance from a point on $l_1$ to $B$ increases as that point moves away from $M$. Hence $d(B, P^\prime)>d(B, P)$, or equivalently $d(T^\prime)>d(T)$.

Thus, we can assume $T$ has two sides of length $a$, that is $a=b>c$.
% As in the previous case, we will replace $A$ with a point $A^\prime$, but in this case we choose $A^\prime$ to be a point on the circle of radius $a$ centered at $C$ instead of on the line $\overleftrightarrow{AB}$. 
%
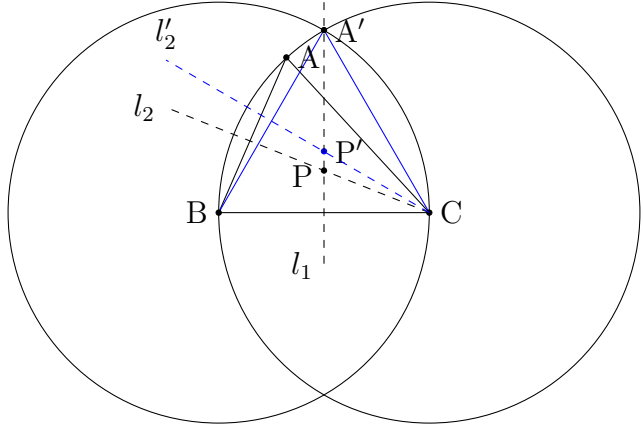
\begin{figure}[h]
	\centering
	\begin{tikzpicture}[scale=.7]
									\draw [color=black] (0,0) -- (4,0) -- (1.285,2.95) -- (0,0);
							
									\draw [color=blue] (0,0) -- (2, 3.47) -- (4, 0);
									\draw[fill] (0,0) circle [radius=0.05];%B
									\draw[fill] (4,0) circle [radius=0.05];%C
									\draw[fill] (1.285,2.95) circle [radius=0.05];%A
									\draw (0,0) circle [radius=4];
									\draw (4, 0) circle [radius=4];
									\draw[fill] (2, 3.47) circle [radius=0.05];%A'
									\draw[fill] (2,.8) circle [radius=0.05];
									\draw[fill, color=blue] (2,1.166) circle [radius=0.05];

									\node [left] at (0,0) {B};
									\node [right] at (4,0) {C};
									\node [right] at (1.285,2.95) {A};
									\node [right] at (2, 3.47) {A$^\prime$};
									\node [left] at (-1,2) {$l_2$};
									\node [above] at (-1,2.9) {$l_2^\prime$};	
									\node [left] at (2, -1) {$l_1$};	
									\node [left] at (2,.7) {P};
									\node [right] at (2,1.166) {P$^\prime$};

									\draw[dashed] (2,4)-- (2,-1);
									\draw[dashed] (4, 0)--(-1, 2);	
									\draw[dashed, color=blue] (4, 0)--(-1, 2.9);
									\end{tikzpicture}	
\caption{Replacing an isosceles triangle $\triangle ABC$ with an equilateral triangle $\triangle A'BC$.} \label{fig:triangle2}
\end{figure}
Consider the circles of radius $a$ centered at $B$ and $C$, see Figure~\ref{fig:triangle2}. When $K\leq 0$, it is easy to see that they intersect at two points, one on each side of $\overleftrightarrow{BC}$. When $K>0$, let $M'$ denote the point on the sphere that is antipodal to $M$. Then $d(B, M)=\frac{a}{2}<a$ and since we assumed that $a<\frac{2\pi}{3}R$, $d(B, M')=\pi R-\frac{a}{2}>\frac{2\pi}{3}R>a$. Note that $M$ and $M'$ both lie on $l_1$, hence there exist two points on $l_1$, one on each side of $\overleftrightarrow{BC}$, whose distance to $B$ is equal to $a$. Since these points lie on $l_1$, they also have distance $a$ to $C$, and hence they lie on the intersection of the two circles.

%They intersect at two points.

Let $A^\prime$  be the intersection point of these two circles on the same side of $\overleftrightarrow{BC}$ as $A$.
% on the circle of radius $a$ centered at $C$ where $d(A', B)=a$;
% Equivalently, $A'$ is a point where $l_1$ intersects this circle. Of the two such points on this circle, we assume $A^\prime$ is the point where the line from $B$ to $A^\prime$ crosses segment $\overline{AC}$.
Again, let $T^\prime$ be the triangle with vertices $A^\prime$, $B$, and $C$. By construction $T^\prime$ is an equilateral triangle with side lengths $a$, hence $b(T^\prime)=\frac{a}{2}=b(T)$. 

Let $l_2'$ be the perpendicular bisector of $\overline{A'B}$. By construction, the angle of $T$ at vertex $C$ is smaller then the angle of $T'$ at vertex $C$. Since these are both isosceles triangles, $l_2$ and $l_2'$ bisect these angles respectively. Hence, the angle formed by $\overline{BC}$ and $l_2$ is smaller then the angle formed by $\overline{BC}$ and $l_2'$. It follows that the point $P$ where $l_2$ intersects $l_1$ is closer to $\overline{BC}$ then the point $P'$ where $l_2'$ intersects $l_2$. As before, this means that $d(B, P')>d(B, P)$. Since $P$ and $P'$ are the circumcenters of $T$ and $T'$, we get that $d(T')>d(T)$.
\end{proof}

\subsection{Persistence of equilateral triangles} 

In this section, we give formulas for the persistence $p(T)$ where $T$ is an equilateral triangle in $M_K$. In general it is possible to give formulas for the persistence of arbitrary triangles in $M_K$ in terms of the side lengths of $T$ and $K$ since $b(T)$ is half the length of the longest side of $T$ and $d(T)$ is the circumradius of $T$. In the general case these formulas are not particularly enlightening, however for equilateral triangles the generalized law of sines (Theorem~\ref{thm:sine-law}) allows us to simplify the formulas considerably.

\begin{theorem} \label{thm:persistence}
	Let $T_{K,a}$ be an equilateral triangle in $M_K$ with side length $a$.
	% If $K=0$, then
	% $$p(T) = \dfrac{2}{\sqrt{3}}.$$
	% If $K>0$, then 
	% $$p(T) = \dfrac{2}{a\sqrt{K}}\sin^{-1}\bigg(\dfrac{2}{\sqrt{3}}\sin\bigg(\dfrac{a\sqrt{K}}{2}\bigg) \bigg)$$
	% and if $K<0$, then 
	% $$p(T) = \dfrac{2}{a\sqrt{-K}}\sinh^{-1}\bigg(\dfrac{2}{\sqrt{3}}\sinh\bigg(\dfrac{a\sqrt{-K}}{2}\bigg) \bigg).$$
\begin{equation*}
  p(T_{K,a}) =
  \begin{cases}
\dfrac{2}{a\sqrt{-K}}\sinh^{-1}\bigg(\dfrac{2}{\sqrt{3}}\sinh\bigg(\dfrac{a\sqrt{-K}}{2}\bigg) \bigg) & \text{if } K<0\\
 \dfrac{2}{\sqrt{3}} & \text{if } K=0\\
\dfrac{2}{a\sqrt{K}}\sin^{-1}\bigg(\dfrac{2}{\sqrt{3}}\sin\bigg(\dfrac{a\sqrt{K}}{2}\bigg) \bigg) & \text{if } K>0\new{.}
  \end{cases}
\end{equation*}
\end{theorem}

\begin{proof}
  Let $T$ be a equilateral triangle in $M_K$ with vertices $A$, $B$, and $C$ and side lengths a.
  % For $K=0$, the circumradius of $T$ is $\frac{a}{\sqrt{3}}$, so $d(T) = \frac{a}{\sqrt{3}}$ and hence $p(T) = \frac{2}{\sqrt{3}}$. For $K\neq 0$, we will use the generalized law of sines.
  % For $K>0$,
  Let $M$ be the midpoint of $AB$, and let $P$ be the circumcenter of $T$.
See Figure~\ref{fig:triangle3}.
Since $P$ is the circumcenter, it is the intersection of the perpendicular bisectors of the sides of $T$. Thus, $\angle AMP = \pi/2$. Moreover, these perpendicular bisector split $T$ into 6 congruent triangles which all contain and surround the vertex $P$. It follows that the angles of these triangles at the vertex $P$ sum to $2\pi$, and since the angles are all congruent we get $\angle APM = \pi/3$. Furthermore, the length of $\overline{AM}$ is $b(T)$ and the length of $\overline{AP}$ is $d(T)$.
	
	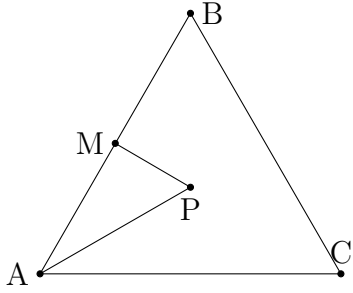
\begin{figure}[h]
	\centering
	\begin{tikzpicture}[scale=.80]
									\draw [color=black] (0,0) -- (5,0) -- (2.5,4.333) -- (0,0);
									\draw[fill] (0,0) circle [radius=0.05];
									\draw[fill] (2.5,4.333) circle [radius=0.05];
									\draw[fill] (1.25,2.17) circle [radius=0.05];
									\draw[fill] (5,0) circle [radius=0.05];
									\draw[fill] (2.5,1.44) circle [radius=0.05];
									\draw [color=black] (0,0) -- (2.5,1.44) -- (1.25,2.17);
									\node [left] at (0,0) {A};
									\node [above] at (5,0) {C};
									\node [right] at (2.5,4.333) {B};
									\node [left] at (1.25,2.17) {M};
									\node [below] at (2.5,1.44) {P};
									\end{tikzpicture}	
\caption{$b(T)=d(A, M)$ and $d(T)=d(A, P)$.} \label{fig:triangle3}
\end{figure}

We apply the generalized law of sines (Theorem~\ref{thm:sine-law}) to the triangle $\Delta AMP$.
For $K=0$, we have 
\begin{equation*}
  \frac{d(T)}{b(T)} = \frac{d(A, P)}{d(A, M)} = \dfrac{\sin(\angle AMP)}{\sin(\angle APM} = \frac{\sin \pi/2}{\sin \pi/3} = \frac{2}{\sqrt{3}}.
\end{equation*}
For $K > 0$,
we have 
	% $$\dfrac{\sqrt{K}\sin(\angle APM)}{\sin(b(T)\sqrt{K})} = \dfrac{\sqrt{K}\sin(\angle AMP)}{\sin(d(T)\sqrt{K})}. $$
	% So
\[
	\dfrac{\sin(d(T)\sqrt{K})}{\sin(b(T)\sqrt{K})} 
%= \dfrac{\sin(\angle AMP)}{\sin(\angle APM}= \dfrac{\sin(\pi/2)}{\sin(\pi/3)}
= \dfrac{2}{\sqrt{3}}.
\]	
	Then 
	$$d(T) = \dfrac{1}{\sqrt{K}}\sin^{-1}\bigg(\dfrac{2}{\sqrt{3}}\sin\bigg(\dfrac{a\sqrt{K}}{2}\bigg) \bigg).$$
	Similarly when $K<0$, 
	% $$\dfrac{\sinh(d(T)\sqrt{-K})}{\sinh(b(T)\sqrt{-K})} = \dfrac{\sin(\angle AMP)}{\sin(\angle APM)} = \dfrac{\text{sin}(\pi/2)}{\text{sin}(\pi/3)} = \dfrac{2}{\sqrt{3}}$$
	% and so 
        \begin{equation*}
          d(T) = \dfrac{1}{\sqrt{-K}}\sinh^{-1}\bigg(\dfrac{2}{\sqrt{3}}\sinh\bigg(\dfrac{a\sqrt{-K}}{2}\bigg) \bigg). \qedhere
        \end{equation*}
      \end{proof}

      From these formulas, \new{one} can easily compute that for any fixed $a$, the function which assigns to $K$ the persistence of an equilateral triangle of side length $a$ in $M_K$ is an increasing \new{and} continuous function. Indeed, the fact that this function converges to $\dfrac{2}{\sqrt{3}}$ as $K\to 0$  is straightforward application of l'H\^{o}pital's rule.
      To get a sense of scale, if $a=1$ then
% $p(T)\approx 1.1063$ in $M_{-5}$, $p(T)\approx 1.1547$ in $M_0$, and $p(T)\approx 1.4050$ in $M_5$.
      the values for $p(T)$ for $K = -2$ ,$-1$, $0$, $1$, and $2$ are approximately $1.1294$, $1.1406$, $1.1547$, $1.1733$, and $1.1996$.

\begin{corollary} \label{cor:increasing}
\new{
  Let $a>0$. Let $p_a(K)$ denote the persistence of an equilateral triangle of side length $a$ in a surface of constant curvature $K$. Then $p_a(K)$ is a continuous and increasing function.}
\end{corollary}

Combining Theorems \ref{thm:isoperimetric} and \ref{thm:persistence}
\new{
  and Corollary~\ref{cor:increasing},
  }
we obtain Theorem~\ref{thm:invertible}.

\section{Estimating curvature using persistence} \label{sec:estimate-K}

In this section, we demonstrate that using the persistent homology of the Vietoris-Rips complex of points sampled on disks of constant curvature we are able to produce good estimates of the curvature.

\subsection{Sampling points uniformly for a unit disk of constant curvature}
\label{sec:sample}

We need to sample points uniformly (with respect to the area measure) from disks of constant curvature with radius one.
See Figure~\ref{fig:sampled-points}.

\subsubsection{Euclidean case}

We start with the Euclidean case, $K=0$.
Consider the disk of radius one centered at the origin. 
Parametrize points on this disk by an angle, $0 \leq \theta < 2\pi$, and a radius, $0 \leq r \leq 1$.
We will sample $\theta$ and $r$ independently. 
For $\theta$, sample uniformly, drawing from the uniform distribution on $[0,2\pi]$.
For $r$, the  probability  of a point lying within the disk of radius $r$ should  equal the  proportion to the area \new{of} that disk relative to the area of  the  disk of radius 1. The area of a disk of radius $r$ equals $\pi r^2$. So the cumulative distribution function of $r$ is given by $$F(r) = \dfrac{\pi r^2}{\pi 1^2} = r^2$$ and the inverse cumulative distribution function is given by
\[
  r = F^{-1}(u) = \sqrt{u}.
\]
So we can sample $u$ uniformly on $[0,1]$ and use $F^{-1}(u)$ to obtain the desired sample of $r$
(see Section~\ref{sec:inversion-sampling}).
% Lastly, we transform our coordinates back to Cartesian coordinates by sending $(\theta ,r)$ to $(r\text{cos}\theta ,r\text{sin}\theta)$. 

% \begin{figure}[h]
% \centering
% \includegraphics[width=.5\linewidth]{Points0}
% \caption{1000 points generated uniformly from a disk of radius 1 in Euclidean space.}
% \end{figure}

%%%%%%%%%%%%%%%%%%%%%%

\begin{figure}[h]
  \centering
  \includegraphics[width=.33\linewidth]{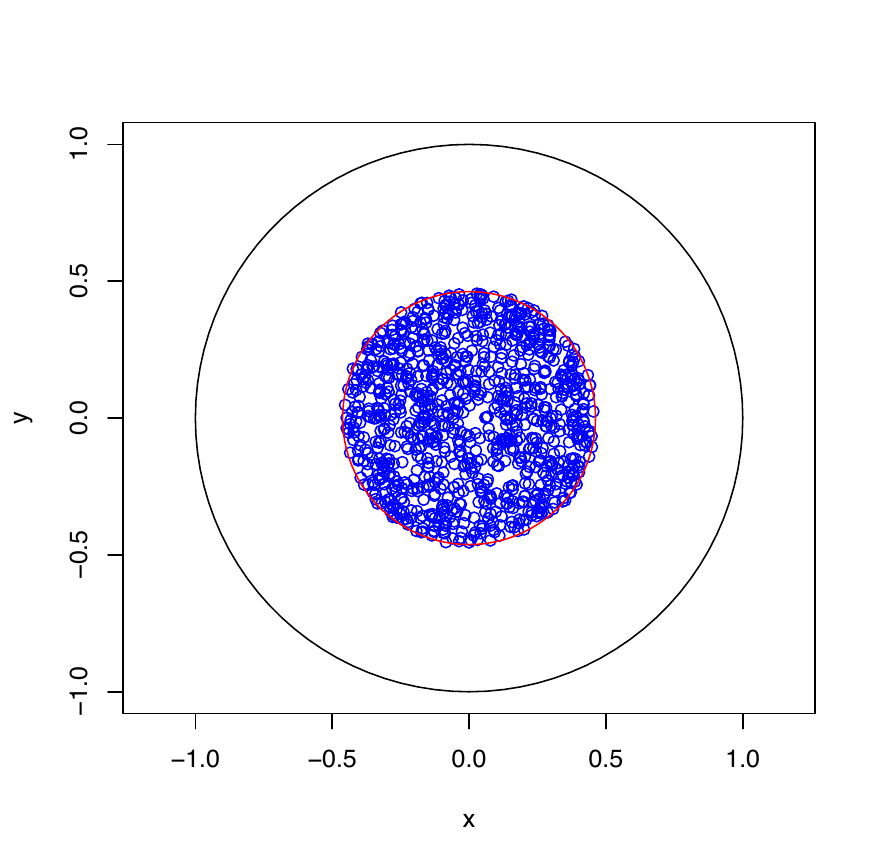}%
  \includegraphics[width=.33\linewidth]{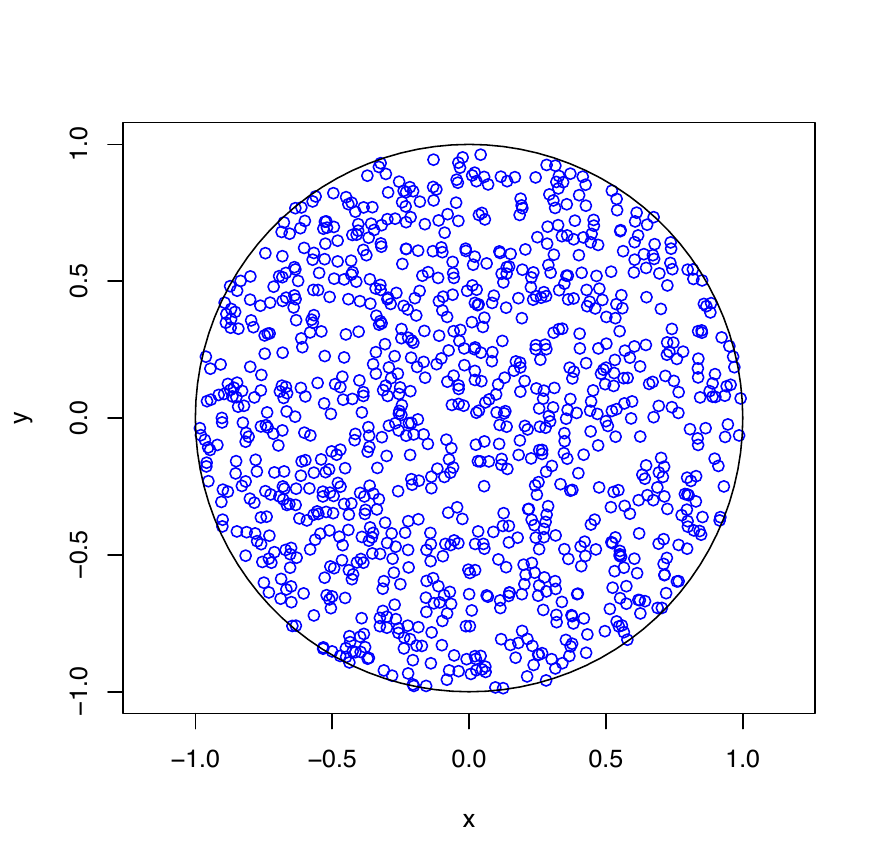}%
  \includegraphics[width=.33\linewidth]{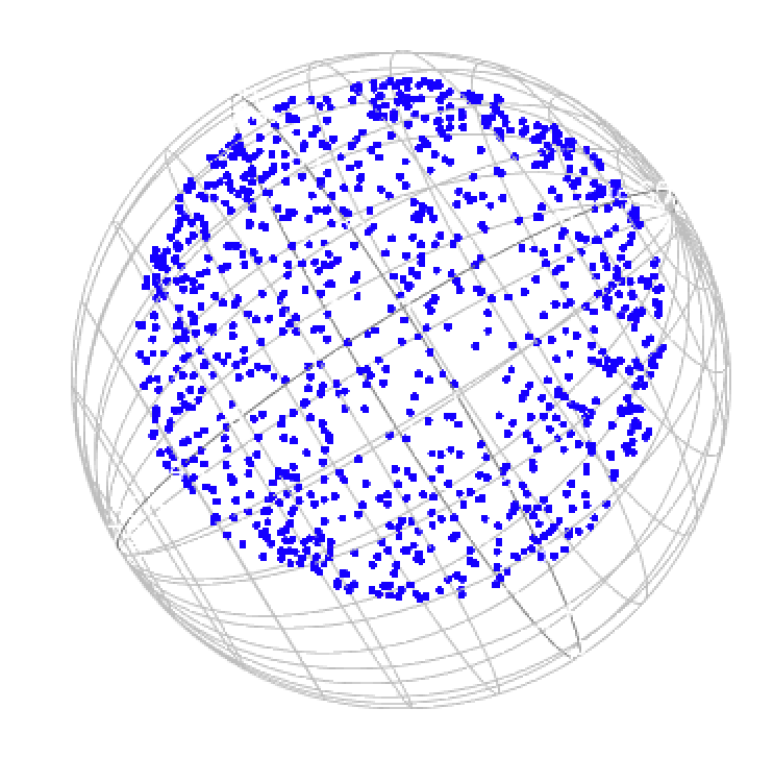}%
  \caption{Plots of $1000$ points sampled independently, with each point sampled uniformly with respect to area for the unit disk on the Poincar\'e disk model of the hyperbolic plane (left), the Euclidean plane (center), and a sphere of radius $1$ (right).}
  \label{fig:sampled-points}
\end{figure}

%%%%%%%%%%%%%%%%%%%%%%%%

\subsubsection{Spherical case}

Next, consider the spherical case, $K>0$.
%Let $R = \frac{1}{\sqrt{K}}$.
  We will assume that $K \leq 2$ which will ensure that we are able to embed a disk with radius one on the upper hemisphere of a sphere with constant curvature $K$.

% Let $S^2_R \subset \R^3$ denote the sphere centered at the origin with radius $R$, where $R = \frac{1}{\sqrt{K}}$. 
% By assumption $R \geq \frac{1}{\sqrt{5}}$, so the distance on $S^2_R$ between antipodal points equals $\pi R \geq \frac{\pi}{\sqrt{5}} \approx 1.4$.
% Consider the disk on $S^2_R$ centered at the north pole $(0,0,R)$ with radius one. 
% We parametrize points in this disk by an an angle, $0 \leq \theta < 2\pi$, in the $xy$-plane and a height $h = R - z$.

Our sampling procedure follows the Euclidean case.
We parametrize points on a disk of radius one with an angle $0 \leq \theta \leq 2\pi$ and a radius $0 \leq r \leq 1$.
We sample $\theta$ and $r$ independently, sampling $\theta$ from uniform distribution on $[0,2\pi]$.
For $r$, the disk of radius $r$ has area $\frac{4\pi}{K} \sin^2(\frac{r\sqrt{K}}{2})$.
So the cumulative distribution for $r$ is given by
\[
  F(r) = \frac{\frac{4\pi}{K} \sin^2(\frac{r\sqrt{K}}{2})}{\frac{4\pi}{K} \sin^2(\frac{1\sqrt{K}}{2})},
\]
and the inverse cumulative distribution is given by
\[
  r = F^{-1}(u) =  \frac{2}{\sqrt{K}}\sin^{-1}\left(\sqrt{u}\sin\left(\frac{\sqrt{K}}{2}\right)\right).
\]
So we can sample $u$ uniformly on $[0,1]$ and use $F^{-1}(u)$ to sample $r$.

\subsubsection{Hyperbolic case}

It remains to consider the hyperbolic case, $K<0$.
%Let $R = \frac{1}{\sqrt{-K}}$.
% Consider the Poincar\'e disk model and the disk centered at the origin with hyperbolic radius one.
We parametrize points on this disk by an angle $0 \leq \theta < 2\pi$ and a  radius $0 \leq r \leq 1$.
As before, we sample $\theta$ and $r$ independently, taking $\theta$ from the uniform distribution on $[0,2\pi]$.
The area of a hyperbolic disk of hyperbolic radius $r$ is given by
$\frac{4\pi}{-K} \sinh^2(\frac{r\sqrt{-K}}{2})$.
Thus the cumulative distribution of $r$ is given by 
\[
  F(r) = \dfrac{\frac{4\pi}{-K} \sinh^2(\frac{r\sqrt{-K}}{2})}{\frac{4\pi}{-K} \sinh^2(\frac{1\sqrt{-K}}{2})}
\]
and the inverse cumulative distribution is given by
\[
  r = F^{-1}(u) = \frac{2}{\sqrt{-K}} \sinh^{-1}\left(\sqrt{u}\sinh\left(\frac{\sqrt{-K}}{2}\right)\right).
\]

\subsection{Average death vectors and average persistence landscapes} \label{sec:average}

For a given curvature $K$, we independently sample $1000$ points from the unit disk in the surface of constant curvature $K$, uniformly with respect to the area measure
(Section~\ref{sec:sample}) and compute the pairwise distances between such points (Section~\ref{sec:distance}).
From this pairwise distance data, we compute the persistent homology of the corresponding Vietoris-Rips complex (Section \ref{sec:sc} and \ref{sec:ph}).
We encode the persistent (reduced) homology in degree $0$ as a death vector and the persistent homology in degree $1$ as a persistence landscape (Section~\ref{sec:ph}).
We then repeat this $100$ times and
average the vectors to obtain an average death vector and average persistence landscape.
See \new{Figures} \ref{fig:adv} and \ref{fig:apl}.

\begin{figure}[h]
\centering
\includegraphics[width=.33\linewidth]{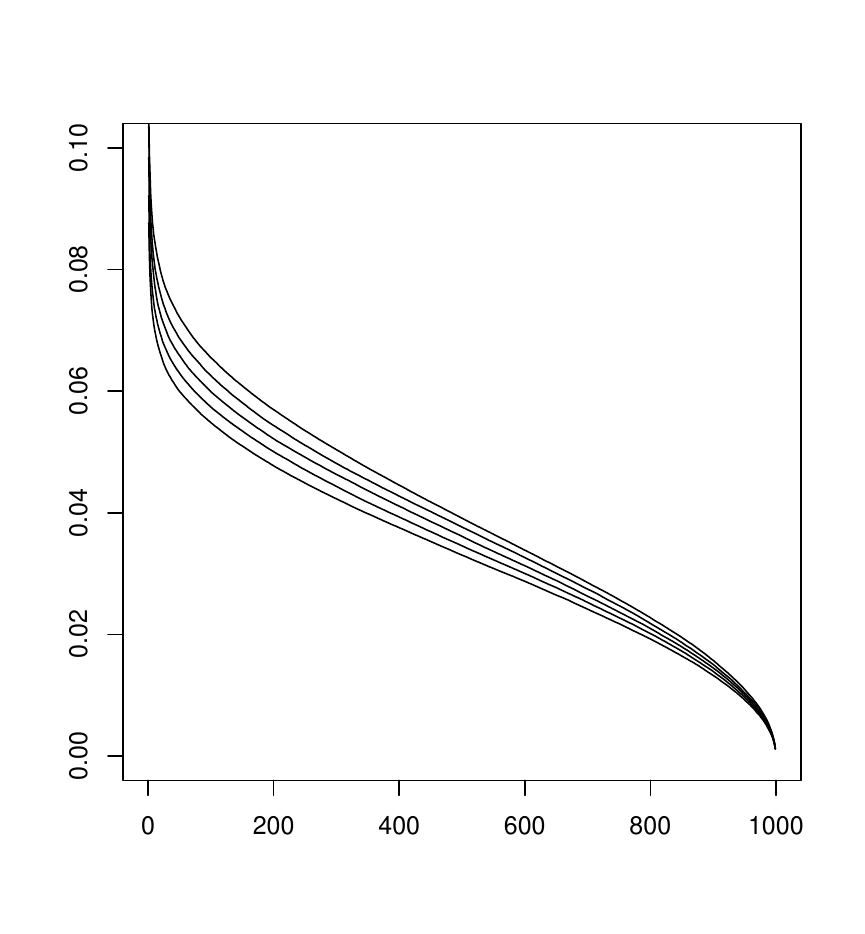}
\caption{The average death vectors from $100$ samples of $1000$ points sampled uniformly from the unit disk in the surface of constant curvature for curvatures $K=-2,-1,0,1,2$ (top to bottom).}
\label{fig:adv}
\end{figure}

\begin{figure}[h]
  \centering
\includegraphics[width=.33\linewidth]{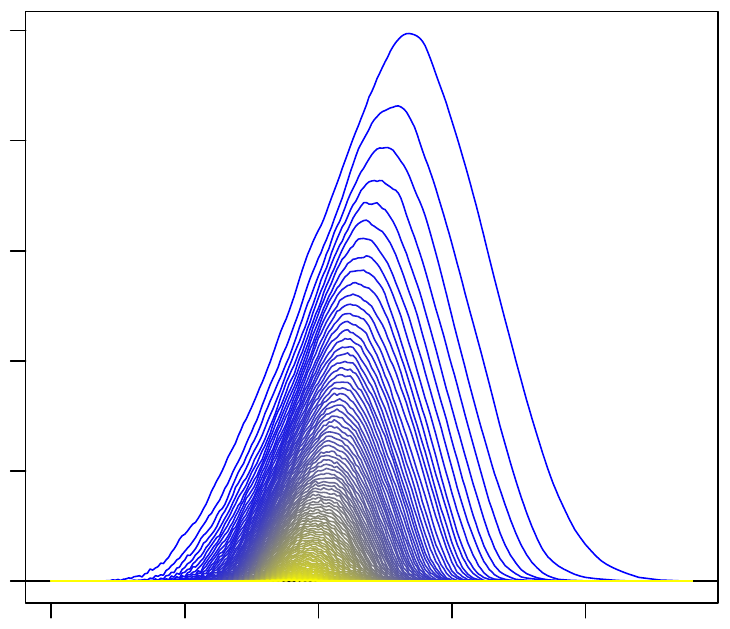}%
\includegraphics[width=.33\linewidth]{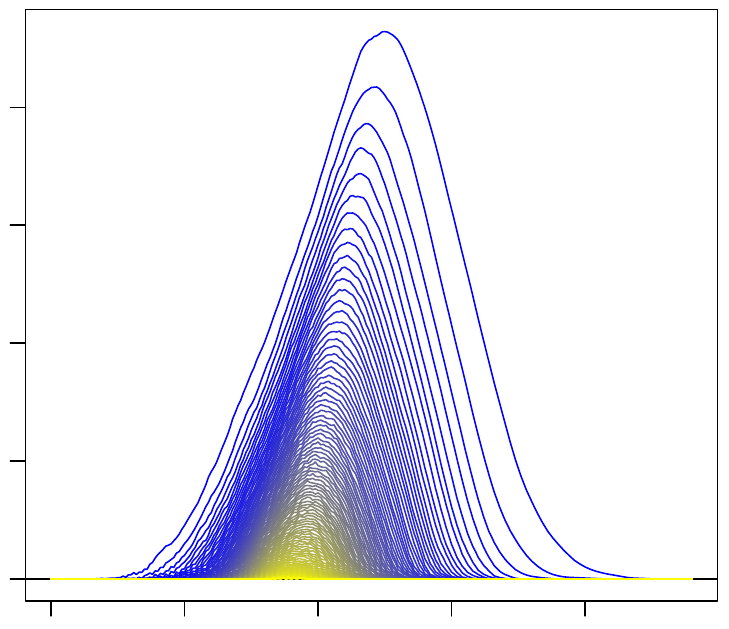}%
\includegraphics[width=.33\linewidth]{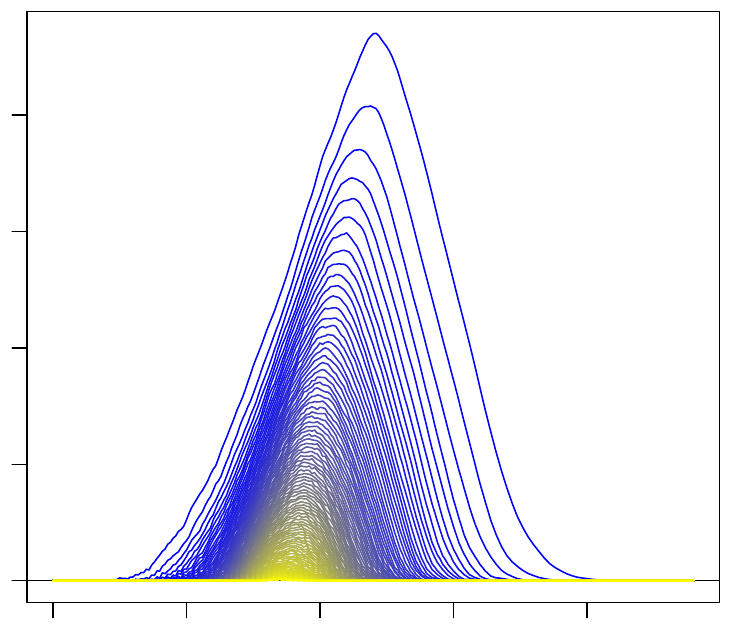}%
  \caption{The average persistence landscapes from $100$ samples of $1000$ points sampled uniformly from the unit disk in the hyperbolic plane (left), the Euclidean plane (center), and a sphere of radius one (right).}
  \label{fig:apl}
\end{figure}

%%%%%%%%%%%%%%%%%%%%%%%%55555555
\subsection{Supervised learning} \label{sec:supervised}

As training data, for each $K \in \{-2,-1.96,-1.92, \ldots, 1.96, 2\}$ we compute the average death vector and average persistence landscape as in Section~\ref{sec:average}.
Call the average death vectors the \emph{$H_0$ training vectors},
call the average persistence landscapes the \emph{$H_1$ training vectors}, and call the concatenations of the average death vectors and the average persistence landscapes the \emph{$H_0$-and-$H_1$ training vectors}.

For testing data, sample 100 curvatures uniformly in $[-2,2]$ and compute their corresponding average death vectors and average persistence landscapes as in Section~\ref{sec:average}.
Call the average death vectors the \emph{$H_0$ testing vectors},
call the average persistence landscapes the \emph{$H_1$ testing vectors}, and call the concatenations of the average death vectors and the average persistence landscapes the \emph{$H_0$-and-$H_1$ testing vectors}.
Now we assume that the testing curvatures are unknown.

\subsubsection{Nearest neighbors} \label{sec:nn}

For each testing vector, find the three nearest training vectors
\new{using the Euclidean distance.}
Estimate the curvature of the testing vector to be the weighted average of the curvatures of the three nearest training vectors,
\new{with the weighting given by the reciprocal of the distance.}
Results are given in Figure~\ref{fig:scatterplot-distance} and  Table~\ref{tab:rmse-distance}.

\subsubsection{Support vector regression}
\label{sec:svr-comp}

We apply support vector regression to the training data to construct a model.
We use a linear loss function and the dot product on the training vectors.
This dot product corresponds to the inner product on the space of persistence landscapes~\cite{Bub}.
We use the ksvm function in the kernlab package~\cite{kernlab} in R with cost $100$ (and $\varepsilon=0$).
%
%As in the previous section, sample 100 curvatures uniformly in $[-2,2]$ and %compute the corresponding testing vectors.
Apply the testing vectors to the linear model computed using support vector regression to estimate the corresponding curvature.
Results are given in Figure~\ref{fig:scatterplot-distance} and  Table~\ref{tab:rmse-distance}.

\subsubsection{Quantile regression}
\label{sec:quantile}

We use the pinball loss function (Section~\ref{sec:svr}) and the dot product on the training $H_0$-and-$H_1$ vectors to construct models that estimate the $\tau$-quantiles for $\tau = 0.05$, $0.5$, and $0.95$.
We use the kqr function in the kernlab package~\cite{kernlab} in R with cost $100$.
Applying the testing vectors to these models we obtain the curves given in Figure~\ref{fig:quantile}.

\begin{figure}[h]
\centering
\includegraphics[width=.23\linewidth]{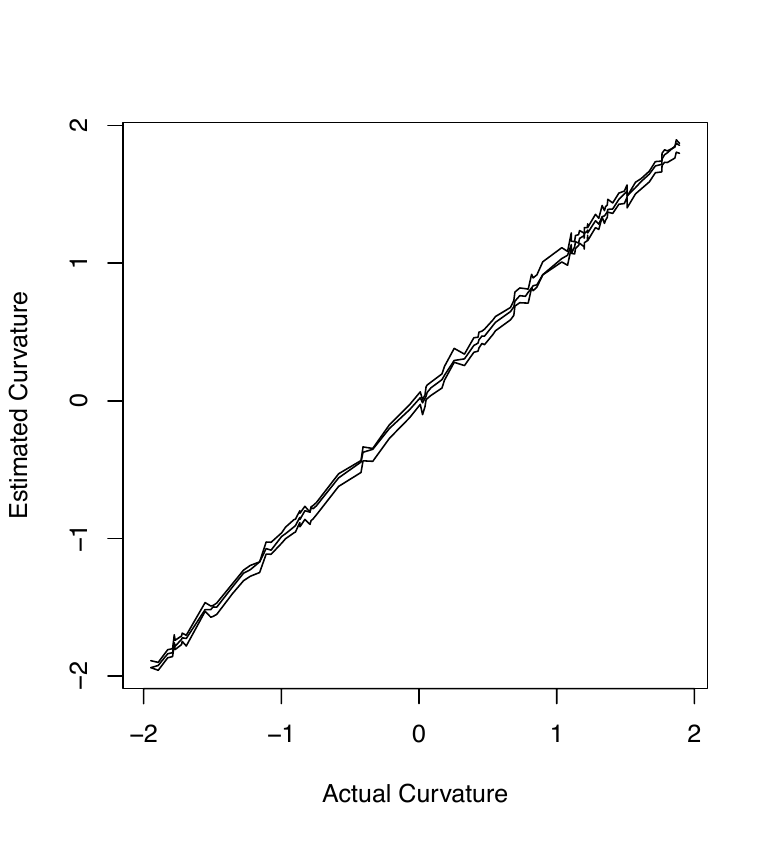}
\caption{Quantile regression for the $H_0$-and-$H_1$ vectors. The middle curve is the estimated median and the bottom and top curves are the estimated fifth and ninety-fifth percentiles, respectively. Testing vectors are obtained for 100 randomly chosen curvatures in $[-2,2]$ and the estimated percentiles for each of these are connected by line segments.}
\label{fig:quantile}
\end{figure}

%%%%%%%%%%%%%%%%%%%%%%%%55555555
\subsection{Unsupervised learning} \label{sec:unsupervised}

Remarkably, we are still able to provide reasonable curvature estimates (up to sign) without any training data. 

Sample 100 curvatures uniformly in $[-2,2]$.
For each of these compute the average death vectors and average persistence landscapes as in Section~\ref{sec:average}. 
Call the average death vectors the \emph{$H_0$ vectors},
call the average persistence landscapes the \emph{$H_1$ vectors}, and call the concatenations of the average death vectors and the average persistence landscapes the \emph{$H_0$-and-$H_1$ vectors}.

For each of these three sets of vectors apply principal components analysis (PCA). The projections onto the first two PCA coordinates for the $H_0$-and-$H_1$ vectors \new{are} given in Figure~\ref{fig:pca}.
Rescale the first principal component axis to [-2,2] and use this to estimate the curvature. 
Results are given in Figure~\ref{fig:scatterplot-distance} and Table~\ref{tab:rmse-distance}. 
Note that with probability $\frac{1}{2}$ this procedure will choose the wrong sign for the estimated curvature.

\begin{figure}[h]
\centering
\includegraphics[width=.23\linewidth]{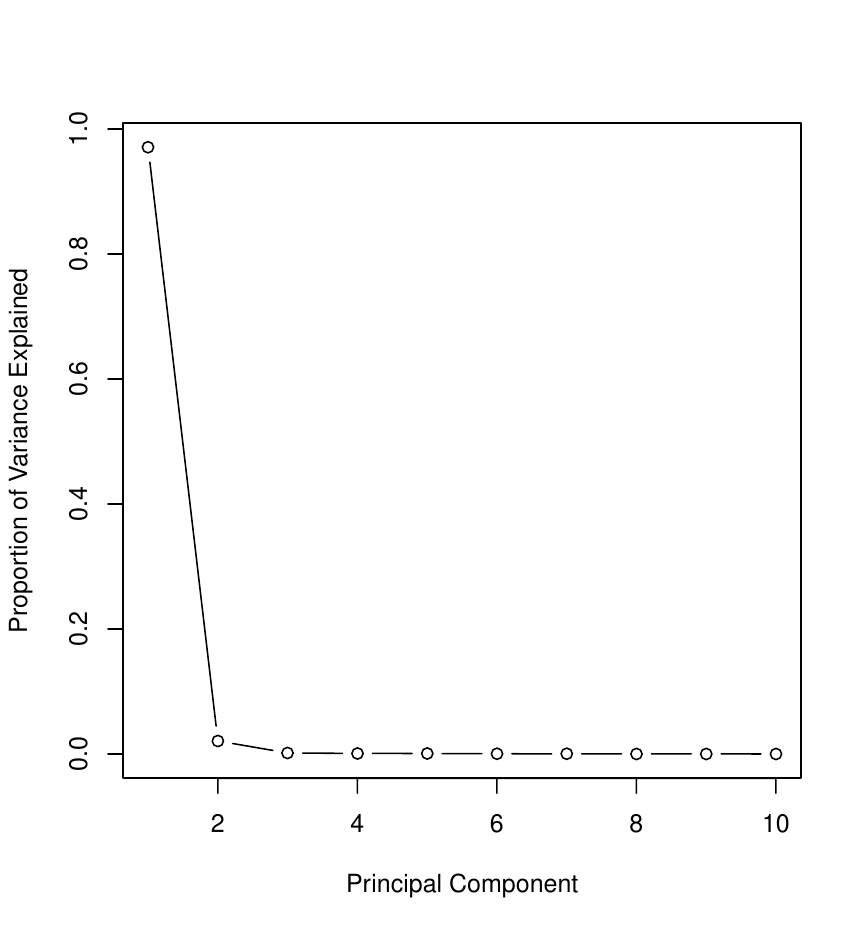}
\includegraphics[width=.23\linewidth]{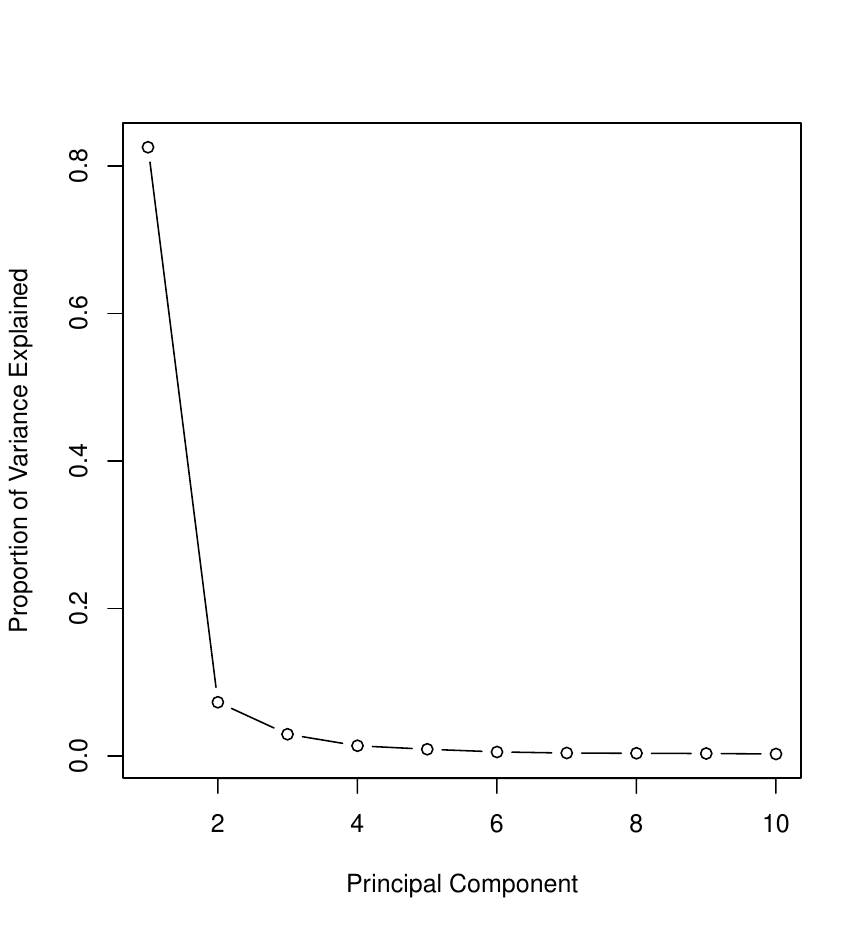}
\includegraphics[width=.23\linewidth]{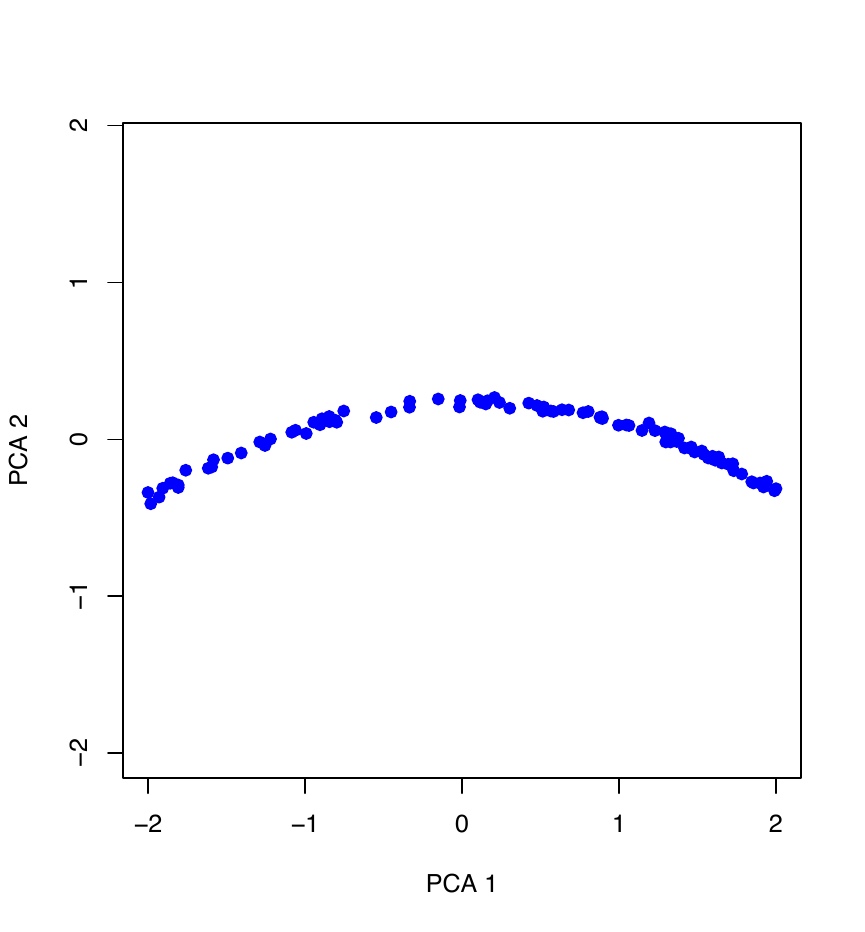}
\includegraphics[width=.23\linewidth]{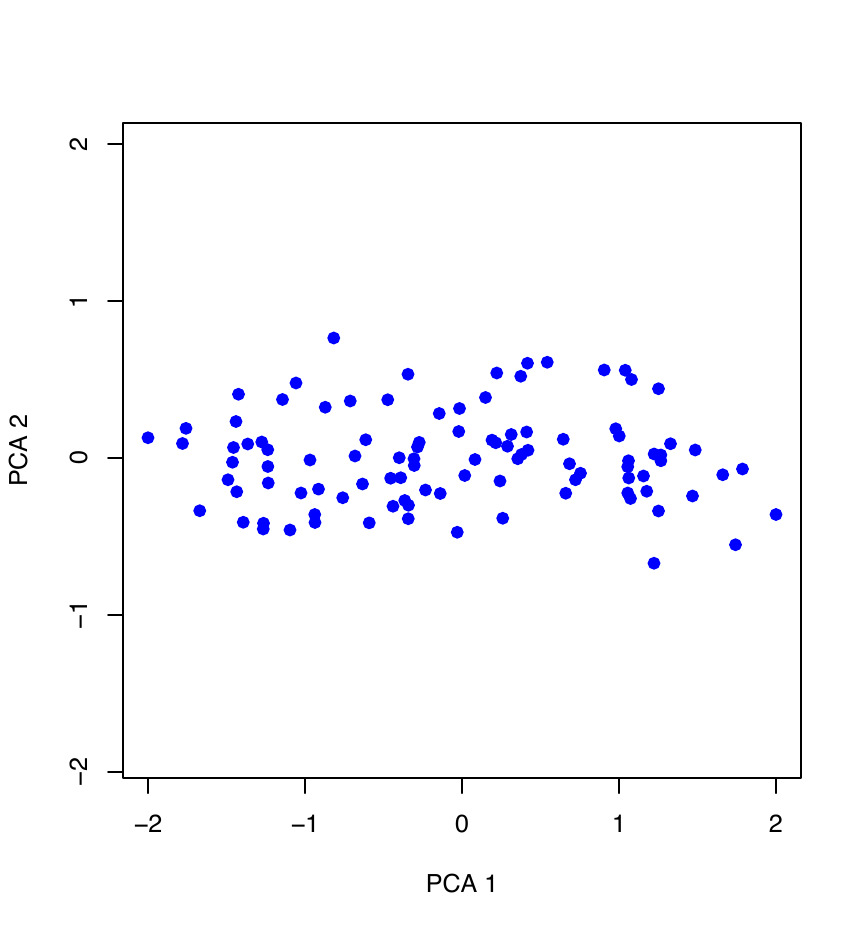}
\caption{Principal components analysis for the $H_0$-and-$H_1$ vectors. The proportion of the variance explained by the first ten principal components using distances (far left) and ordinals (center left).  Projection of the vectors onto the first two principal components using distances (center right) and ordinals (far right).}
\label{fig:pca}
\end{figure}

\begin{table}[h]
\caption{The root mean squared errors of the estimated curvature using pairwise distances.}
\centering
\begin{tabular}{|l|l|c|c|c|}
  \hline
  %{ | X[l] | X[c] | X[r] | } \hline
  & & $H_0$ & $H_1$ & $H_0$-and-$H_1$ \\
  \hline
  Supervised learning & Nearest neighbors & 0.032 & 0.070 & 0.056\\
  & Support Vector Regression & 0.027 & 0.038 & 0.017\\
  \hline
  Unsupervised learning & First Principal Component & 0.091 & 0.139 & 0.128 \\
  \hline
\end{tabular}
\label{tab:rmse-distance}                                              
\end{table}

% The death vector only accounts for the persistence of connected components which is only dependent on the distance between two points in our case. Let $x$ and $y$ be two points in a sample of curvature $K$. Then the distance between them is dependent on $K$. But there is no reason two points $x'$ and $y'$ cannot be sampled with curvature $K'$ such that $d(x,y) = d(x',y')$. In other words the death vector is dependent on the underlying curvature, but curvature is not necessarily dependent on the death vector. This causes $M_0$ to be perform poorly as the above figure shows. 

% On the other hand model $M_0$ works well. The MSE for the testing data for $M_1$ is below .1. Model $M_1$ confirms our intuition in the previous section that there is a connection between the persistence in degree 1 and curvature. Model $M_{0,1}$ adds a little to $M_1$ but not much, meaning that the death vector does not add significantly more information that degree 1 persistence landscapes. 

%%%%%%%%%%%%%%%%%%%%%%%%
\subsection{Using ordinals of sorted pairwise distances}
\label{sec:rank}

In this section we show that our methods do not depend on differences in distributions of the pairwise distances. 

In neuroscience~\cite{Giusti:2015}, certain observed correlations are believed to be given by an unknown monotonic function on an underlying distance in the relevant stimulus space.
Therefore, the distance data should only be used up to monotone transformations.
This can be done by replacing scalar values with ordinal values.

\begin{figure}[h]
\centering
\includegraphics[width=.3\linewidth]{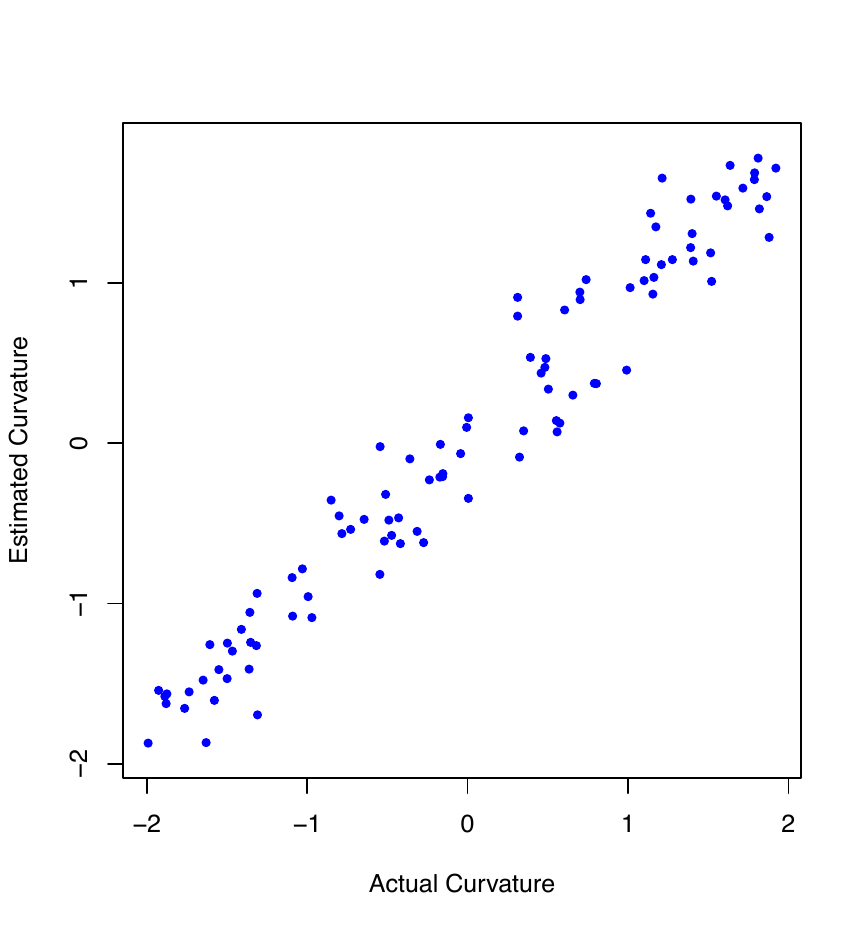}
\includegraphics[width=.3\linewidth]{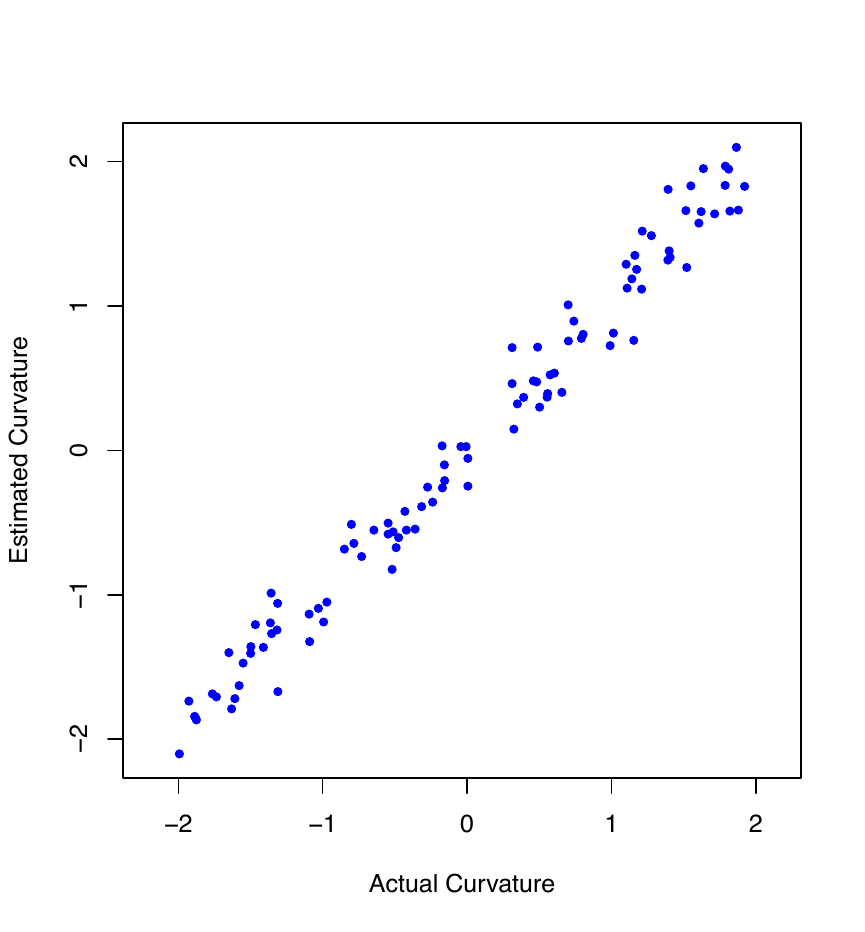}
\includegraphics[width=.3\linewidth]{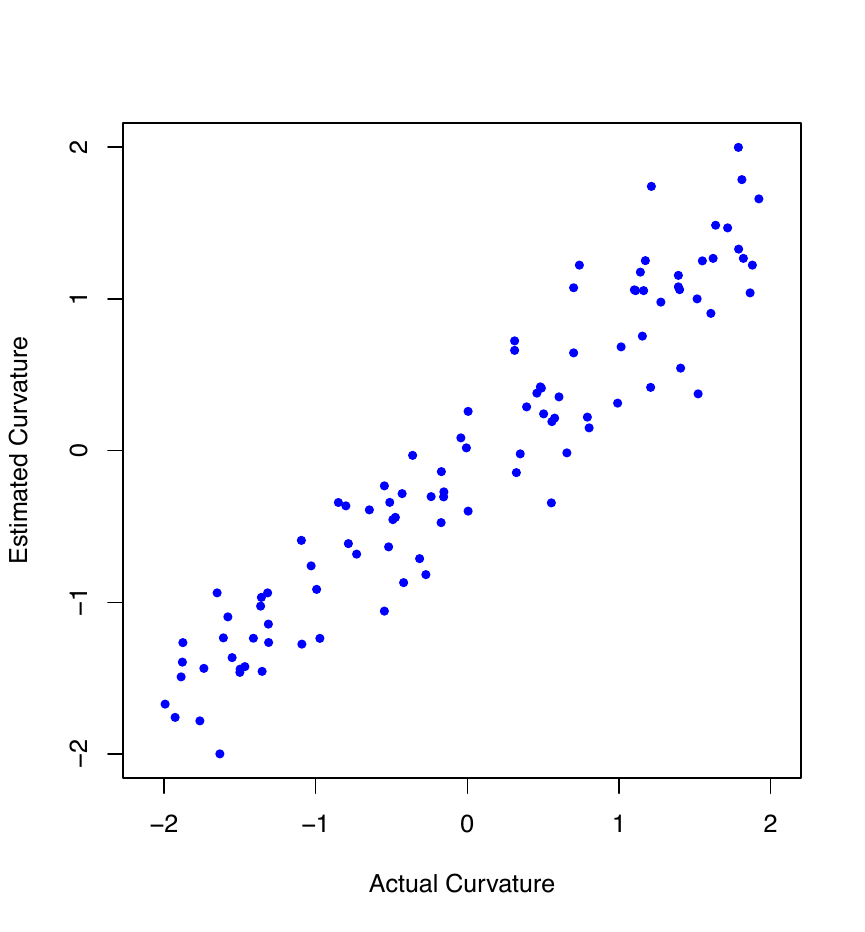}
\caption{Plots showing actual curvature and estimated curvature using  $H_0$ and $H_1$ from ordinal distance data, for nearest neighbors (left), support vector regression (center), and the first principal component (right).}
\label{fig:scatterplot-rank}
\end{figure}

\begin{table}[h]
\caption{The root mean squared errors of the estimated curvatures upon replacing distances with their ordinal numbers.}
\centering
\begin{tabular}{|l|l|c|c|c|}
  \hline
  %{ | X[l] | X[c] | X[r] | } \hline
  & & $H_0$ & $H_1$ & $H_0$-and-$H_1$ \\
  \hline
  Supervised learning & Nearest neighbors & 0.631 & 0.260 & 0.262\\
  & Support Vector Regression & 0.541 & 0.171 & 0.171\\
  \hline
  Unsupervised learning & First Principal Component & 0.615 & 0.393 &  0.392\\
  \hline
\end{tabular}
\label{tab:rmse-rank}                                              
\end{table}

We sort the nonzero pairwise distances and replace them with their corresponding ordinal numbers. 
Thus, for each curvature, the set of nonzero pairwise distances is the set $\{1,2,3,\ldots,\binom{m}{2}\}$.
We redo all of the computations in Sections \ref{sec:supervised} and \ref{sec:unsupervised} in this setting.
For nearest neighbors we use the five nearest neighbors.
For support vector regression we use cost 10 and an $\varepsilon$-insensitive linear loss function with $\varepsilon = 1$ for $H_0$, $\varepsilon=0.2$ for $H_1$, and $\varepsilon=0.2$ for $H_0$ and $H_1$.
We choose different hyper-parameters to avoid over-fitting due to the greater variance in the data.
The results are given in Figure~\ref{fig:scatterplot-rank} and Table~\ref{tab:rmse-rank}.

\subsection*{Acknowledgments}

This research was partially supported by the Southeast Center for Mathematics and Biology, an NSF-Simons Research Center for Mathematics of Complex Biological Systems, under National Science Foundation Grant No. DMS-1764406 and Simons Foundation Grant No. 594594.
This material is based upon work supported by, or in part by, the Army Research Laboratory and the Army Research Office under contract/grant number
W911NF-18-1-0307.
We would also like to thank the anonymous referees for their helpful comments.

% \bibliography{../biblio}

\begin{thebibliography}{10}

\bibitem{Adamaszek:2017}
Micha\l Adamaszek and Henry Adams.
\newblock The {V}ietoris-{R}ips complexes of a circle.
\newblock {\em Pacific J. Math.}, 290(1):1--40, 2017.

\bibitem{Armstrong:1983}
Mark~Anthony Armstrong.
\newblock {\em Basic topology}.
\newblock Undergraduate Texts in Mathematics. Springer-Verlag, New York-Berlin,
  1983.

\bibitem{Beardon:1995}
Alan~F. Beardon.
\newblock {\em The geometry of discrete groups}, volume~91 of {\em Graduate
  Texts in Mathematics}.
\newblock Springer-Verlag, New York, 1995.

\bibitem{bendich:brain-artery}
Paul Bendich, J.~S. Marron, Ezra Miller, Alex Pieloch, and Sean Skwerer.
\newblock Persistent homology analysis of brain artery trees.
\newblock {\em Ann. Appl. Stat.}, 10(1):198--218, 03 2016.

\bibitem{Bobrowski:2017}
Omer Bobrowski, Matthew Kahle, and Primoz Skraba.
\newblock Maximally persistent cycles in random geometric complexes.
\newblock {\em Ann. Appl. Probab.}, 27(4):2032--2060, 2017.

\bibitem{Bub}
Peter Bubenik.
\newblock Statistical topological data analysis using persistence landscapes.
\newblock {\em J. Mach. Learn. Res.}, 16:77--102, 2015.

\bibitem{bubenikDlotko}
Peter Bubenik and Pawel Dlotko.
\newblock A persistence landscapes toolbox for topological statistics.
\newblock {\em Journal of Symbolic Computation}, 78:91 -- 114, 2017.

\bibitem{bubenikKim:parametric}
Peter Bubenik and Peter~T. Kim.
\newblock A statistical approach to persistent homology.
\newblock {\em Homology, Homotopy Appl.}, 9(2):337--362, 2007.

\bibitem{Chavel:2006}
Isaac Chavel.
\newblock {\em Riemannian geometry}, volume~98 of {\em Cambridge Studies in
  Advanced Mathematics}.
\newblock Cambridge University Press, Cambridge, second edition, 2006.

\bibitem{cdso:geometric}
Fr{{\'e}}d{{\'e}}ric Chazal, Vin de~Silva, and Steve Oudot.
\newblock Persistence stability for geometric complexes.
\newblock {\em Geom. Dedicata}, 173:193--214, 2014.

\bibitem{Chazal:2015b}
Fr{\'e}d{\'e}ric Chazal, Brittany~Terese Fasy, Fabrizio Lecci, Bertrand Michel,
  Alessandro Rinaldo, and Larry Wasserman.
\newblock Subsampling methods for persistent homology.
\newblock In {\em Proceedings of the 32nd International Conference on Machine
  Learning, Lille, France}, volume~37. JMLR: W\&CP, 2015.

\bibitem{Chazal:2015c}
Fr{{\'e}}d{{\'e}}ric Chazal, Brittany~Terese Fasy, Fabrizio Lecci, Alessandro
  Rinaldo, and Larry Wasserman.
\newblock Stochastic convergence of persistence landscapes and silhouettes.
\newblock {\em J. Comput. Geom.}, 6(2):140--161, 2015.

\bibitem{Coxeter}
H.S.M. Coxeter.
\newblock {\em Non-Euclidean Geometry}, chapter~12.
\newblock Mathematical Association of America, 1998.

\bibitem{Devroye}
Luc Devroye.
\newblock {\em Non-Uniform Random Variate Generation}.
\newblock Springer, 1986.

\bibitem{Dlotko:2016}
Pawe{\l} D{\l}otko and Thomas Wanner.
\newblock Topological microstructure analysis using persistence landscapes.
\newblock {\em Physica D: Nonlinear Phenomena}, 334:60 -- 81, 2016.

\bibitem{Donato:2016}
Irene Donato, Matteo Gori, Marco Pettini, Giovanni Petri, Sarah De~Nigris,
  Roberto Franzosi, and Francesco Vaccarino.
\newblock Persistent homology analysis of phase transitions.
\newblock {\em Phys. Rev. E}, 93:052138, May 2016.

\bibitem{Edelsbrunner:2014}
Herbert Edelsbrunner.
\newblock {\em A short course in computational geometry and topology}.
\newblock Springer Briefs in Applied Sciences and Technology. Springer, Cham,
  2014.

\bibitem{edelsbrunnerHarer:book}
Herbert Edelsbrunner and John~L. Harer.
\newblock {\em Computational topology}.
\newblock American Mathematical Society, Providence, RI, 2010.

\bibitem{Gameiro:2015b}
Marcio Gameiro, Yasuaki Hiraoka, Shunsuke Izumi, Miroslav Kramar, Konstantin
  Mischaikow, and Vidit Nanda.
\newblock A topological measurement of protein compressibility.
\newblock {\em Jpn. J. Ind. Appl. Math.}, 32(1):1--17, 2015.

\bibitem{Gidea:2018}
Marian Gidea and Yuri Katz.
\newblock Topological data analysis of financial time series: landscapes of
  crashes.
\newblock {\em Phys. A}, 491:820--834, 2018.

\bibitem{Giusti:2015}
Chad Giusti, Eva Pastalkova, Carina Curto, and Vladimir Itskov.
\newblock Clique topology reveals intrinsic geometric structure in neural
  correlations.
\newblock {\em Proc. Natl. Acad. Sci. USA}, 112(44):13455--13460, 2015.

\bibitem{Hiraoka:2016}
Yasuaki Hiraoka, Takenobu Nakamura, Akihiko Hirata, Emerson~G. Escolar, Kaname
  Matsue, and Yasumasa Nishiura.
\newblock Hierarchical structures of amorphous solids characterized by
  persistent homology.
\newblock {\em Proceedings of the National Academy of Sciences},
  113(26):7035--7040, 2016.

\bibitem{Ichinomiya:2017}
Takashi Ichinomiya, Ippei Obayashi, and Yasuaki Hiraoka.
\newblock Persistent homology analysis of craze formation.
\newblock {\em Phys. Rev. E}, 95:012504, Jan 2017.

\bibitem{Jiang:2018}
Fei Jiang, Takeshi Tsuji, and Tomoyuki Shirai.
\newblock Pore geometry characterization by persistent homology theory.
\newblock {\em Water Resources Research}, 54(6):4150--4163, 2018.

\bibitem{Kanari2017}
Lida Kanari, Pawe{\l} D{\l}otko, Martina Scolamiero, Ran Levi, Julian
  Shillcock, Kathryn Hess, and Henry Markram.
\newblock A topological representation of branching neuronal morphologies.
\newblock {\em Neuroinformatics}, Oct 2017.

\bibitem{kernlab}
Alexandros Karatzoglou, Alex Smola, Kurt Hornik, and Achim Zeileis.
\newblock kernlab -- an {S4} package for kernel methods in {R}.
\newblock {\em Journal of Statistical Software}, 11(9):1--20, 2004.

\bibitem{Kondic:2016}
L.~Kondic, M.~Kram\'ar, Luis~A. Pugnaloni, C.~Manuel Carlevaro, and
  K.~Mischaikow.
\newblock Structure of force networks in tapped particulate systems of disks
  and pentagons. ii. persistence analysis.
\newblock {\em Phys. Rev. E}, 93:062903, Jun 2016.

\bibitem{giseon:maltose}
Violeta Kovacev-Nikolic, Peter Bubenik, Dragan Nikoli{{\'c}}, and Giseon Heo.
\newblock Using persistent homology and dynamical distances to analyze protein
  binding.
\newblock {\em Stat. Appl. Genet. Mol. Biol.}, 15(1):19--38, 2016.

\bibitem{Kramar:2014}
Miroslav Kram{{\'a}}r, Arnaud Goullet, Lou Kondic, and Konstantin Mischaikow.
\newblock Quantifying force networks in particulate systems.
\newblock {\em Phys. D}, 283:37--55, 2014.

\bibitem{Kramar:2016}
Miroslav Kram{\'a}r, Rachel Levanger, Jeffrey Tithof, Balachandra Suri, Mu~Xu,
  Mark Paul, Michael~F. Schatz, and Konstantin Mischaikow.
\newblock Analysis of kolmogorov flow and rayleigh--b{\'e}nard convection using
  persistent homology.
\newblock {\em Physica D: Nonlinear Phenomena}, 334:82 -- 98, 2016.

\bibitem{Memoli:2011b}
Facundo M{\'e}moli.
\newblock Gromov-{W}asserstein distances and the metric approach to object
  matching.
\newblock {\em Found. Comput. Math.}, 11(4):417--487, 2011.

\bibitem{oudot:book}
Steve~Y. Oudot.
\newblock {\em Persistence theory: from quiver representations to data
  analysis}, volume 209 of {\em Mathematical Surveys and Monographs}.
\newblock American Mathematical Society, Providence, RI, 2015.

\bibitem{Patrangenaru:2018}
Vic Patrangenaru, Peter Bubenik, Robert~L. Paige, and Daniel Osborne.
\newblock Challenges in topological object data analysis.
\newblock {\em Sankhya A}, Sep 2018.

\bibitem{Robins:2016}
Vanessa Robins and Katharine Turner.
\newblock Principal component analysis of persistent homology rank functions
  with case studies of spatial point patterns, sphere packing and colloids.
\newblock {\em Phys. D}, 334:99--117, 2016.

\bibitem{Schweinhart:2018a}
Benjamin Schweinhart.
\newblock The persistent homology of random geometric complexes on fractals.
\newblock arXiv:1808.02196 [math.PR], 08 2018.

\bibitem{Smola:2004}
Alex~J. Smola and Bernhard Sch\"{o}lkopf.
\newblock A tutorial on support vector regression.
\newblock {\em Stat. Comput.}, 14(3):199--222, 2004.

\bibitem{Steinwart:2008}
Ingo Steinwart and Andreas Christmann.
\newblock {\em Support vector machines}.
\newblock Information Science and Statistics. Springer, New York, 2008.

\bibitem{Xia:2015}
Kelin Xia, Xin Feng, Yiying Tong, and Guo~Wei Wei.
\newblock Persistent homology for the quantitative prediction of fullerene
  stability.
\newblock {\em Journal of Computational Chemistry}, 36(6):408--422, 2015.

\end{thebibliography}
% \bibliographystyle{plain}

\end{document}